\documentclass[journal]{IEEEtran}

\usepackage{graphicx}

\usepackage{amsmath,amsfonts,amssymb,amsthm} 
\usepackage[ruled,vlined]{algorithm2e}
\usepackage{booktabs}
\usepackage{array}
\usepackage{makecell}
\usepackage[caption=false,font=normalsize,labelfont=sf,textfont=sf]{subfig}
\usepackage{textcomp}
\usepackage{stfloats}
\usepackage{url}
\usepackage{verbatim}
\usepackage{cite}
\usepackage{xcolor}
\usepackage{balance}
\usepackage{glossaries}
\usepackage{enumitem}

\def\BibTeX{{\rm B\kern-.05em{\sc i\kern-.025em b}\kern-.08em
    T\kern-.1667em\lower.7ex\hbox{E}\kern-.125emX}}
    
\DeclareMathOperator{\Tr}{Tr}
\newcommand\numberthis{\addtocounter{equation}{1}\tag{\theequation}} 
\DeclareMathOperator*{\argmax}{arg\,max}
\DeclareMathOperator*{\argmin}{arg\,min}

\begin{document}
\newcommand{\GM}[1]{
{\color{blue} {\textbf{GM: }}
    #1
}}
\newcommand{\GMcomment}[1]{
{\footnotesize {\color{blue} {[\textbf{GM: }}
    #1
    ]
}}}


\newcommand{\defeq}{\mathrel{:=}}
\newcommand{\bigO}{\mathcal{O}}

\newcommand{\Diag}{\text{diag}}
\newcommand{\Vect}{\text{vec}}
\newcommand{\Transpose}[1]{{#1}^\mathrm{T}}
\newcommand{\TransposeSign}{\mathrm{T}}
\newcommand{\Hermitian}[1]{{#1}^{\dagger}}
\newcommand{\HermitianSign}{\dagger}
\newcommand{\MPInverse}[1]{{#1}^{+}}
\newcommand{\Inverse}[1]{{#1}^{-1}}
\newcommand{\Frob}{\mathrm{F}}
\newcommand{\Row}{\mathbf{r}}
\newcommand{\Col}{\mathbf{c}}
\newcommand{\Comp}{\bot}
\newcommand{\SqrtMat}[1]{{#1}^{1/2}}
\newcommand{\ElemProd}{\odot}
\newcommand{\Kron}{\otimes}
\newcommand{\Vectorize}[1]{\text{vec}({#1})}
\newcommand{\CandVal}[1]{\widetilde{#1}} 
\newcommand{\EstVal}[1]{\widehat{#1}} 
\newcommand{\RdmVal}[1]{\overline{#1}} 

\newcommand{\Rset}{\mathbb{R}}
\newcommand{\Cset}{\mathbb{C}}


\newcommand{\SRxInd}{m}
\newcommand{\SRxTh}{$m \textsuperscript{th}$ }
\newcommand{\SRxAll}{M}
\newcommand{\Aoa}{\theta}
\newcommand{\AoaVec}{\boldsymbol \Theta}
\newcommand{\AoaSteerVec}{\mathbf{a}}
\newcommand{\AoaSteerMat}{\mathbf{A}}

\newcommand{\SubInd}{q}
\newcommand{\SubTh}{$q \textsuperscript{th}$ }
\newcommand{\SubAll}{Q}
\newcommand{\Range}{\tau}
\newcommand{\RangeVec}{\bolsymbol \tau}
\newcommand{\RangeSteerVec}{\mathbf{t}}
\newcommand{\RangeSteerMat}{\mathbf{T}}

\newcommand{\TgtInd}{k}
\newcommand{\TgtTh}{$k \textsuperscript{th}$ }
\newcommand{\TgtAll}{K}
\newcommand{\TgtPosVec}{\mathbf{x}}
\newcommand{\TgtPosMat}{\mathbf{X}}
\newcommand{\TgtPosX}{x}
\newcommand{\TgtPosY}{y}
\newcommand{\TgtVelVec}{\dot{\mathbf{x}}}
\newcommand{\TgtVelMat}{\dot{\mathbf{X}}}
\newcommand{\TgtVelX}{\dot{x}}
\newcommand{\TgtVelY}{\dot{y}}

\newcommand{\ObsInd}{d}
\newcommand{\ObsTh}{$d \textsuperscript{th}$ }
\newcommand{\ObsAll}{D}
\newcommand{\Doppler}{f}
\newcommand{\DopplerVec}{\boldsymbol \Omega}
\newcommand{\DopplerSteerVec}{\mathbf{f}}
\newcommand{\DopplerSteerMat}{\mathbf{F}}

\newcommand{\RpInd}{p}
\newcommand{\RpTh}{$p \textsuperscript{th}$ }
\newcommand{\RpAll}{P} 

\newcommand{\ChnlVec}{\mathbf{h}}
\newcommand{\ChnlMat}{\mathbf{H}}
\newcommand{\MIMO}[1]{{#1}^\text{MIMO}}
\newcommand{\DiagVec}{\mathbf{d}}
\newcommand{\DiagMat}{\mathbf{D}}

\newcommand{\CoefVal}{\beta}
\newcommand{\CoefVec}{\boldsymbol \beta}
\newcommand{\CoefMat}{\mathbf{B}}
\newcommand{\CoefLength}{M}
\newcommand{\CoefVecbis}{\mathbf{b}}
\newcommand{\PathVal}{\alpha}
\newcommand{\PathVec}{\boldsymbol \alpha}

\newcommand{\NoiseVal}{n}
\newcommand{\NoiseVec}{\mathbf{n}}
\newcommand{\NoiseMat}{\mathbf{N}}
\newcommand{\NoiseVar}{\sigma^2}
\newcommand{\NoiseVarVec}{\boldsymbol \sigma^2}
\newcommand{\NoiseSD}{\sigma}
\newcommand{\NoiseSDMM}{\sigma_m}
\newcommand{\NoiseSDMMp}{\sigma_{p,m}}
\newcommand{\SNR}{\mathrm{SNR}}

\newcommand{\SteerVec}{\mathbf{a}}
\newcommand{\SteerMat}{\mathbf{A}}
\newcommand{\Identity}{\mathbf{I}}
\newcommand{\ProjMat}{\mathbf{P}}
\newcommand{\CompProjMat}{\mathbf{P}^{\perp}}
\newcommand{\SCovMat}{\mathbf{R}}

\newcommand{\Pdf}{f}
\newcommand{\MeanVal}{\mu}
\newcommand{\MeanVec}{\boldsymbol \mu}
\newcommand{\CovMat}{\boldsymbol \Gamma}
\newcommand{\RdmVec}{\mathbf{z}}
\newcommand{\RdmMat}{\mathbf{Z}}
\newcommand{\CN}{\mathcal{CN}} 

\newcommand{\ParamVec}{\boldsymbol \gamma}
\newcommand{\LLFun}{\mathcal{L}}
\newcommand{\GridAll}{N}
\newcommand{\GridInd}{n}
\newcommand{\GridTh}{$n \textsuperscript{th}$ }
\newcommand{\Grid}{\mathcal{G}}
\newcommand{\FFTMat}{\mathbf{F}}
\newcommand{\FFTVec}{\mathbf{f}}
\newcommand{\FFTVal}{f}
\newcommand{\Weigts}{w}

\newcommand{\PseudoSpec}{\mathcal{J}}
\newcommand{\PseudoVec}{\mathbf{J}}

\newcommand{\SigSpaceVec}{\mathbf{u}}
\newcommand{\SigSpaceMat}{\mathbf{U}}
\newcommand{\SigEigVal}{\lambda}
\newcommand{\SigEigVec}{\boldsymbol \lambda}
\newcommand{\SigEigMat}{\boldsymbol \Lambda}

\newcommand{\NoiseSpaceVec}{\mathbf{g}}
\newcommand{\NoiseSpaceMat}{\mathbf{G}}
\newcommand{\NoiseEigVal}{\varrho}
\newcommand{\NoiseEigVec}{\boldsymbol \varrho}
\newcommand{\NoiseEigMat}{\boldsymbol \Sigma}

\newcommand{\GreedySet}[2]{{#1}_{#2}}

\newcommand{\WaveNumber}{k_c}
\newcommand{\SubSpacing}{\Delta_f}
\newcommand{\ArraySpacing}{\Delta_d}
\newcommand{\Bandwith}{B}
\newcommand{\TimeFrame}{T}
\newcommand{\CpLen}{T_{\mathrm{cp}}}
\newcommand{\SymbolVal}{s}
\newcommand{\SymbolVec}{\mathbf{s}}
\newcommand{\SymbolMat}{\mathbf{S}}
\newcommand{\MappingMat}{\boldsymbol \Xi}
\newcommand{\SRxSigVec}{\mathbf{y}}
\newcommand{\SRxSigMat}{\mathbf{Y}}

\newcommand{\CorrRatio}{\mathcal{S}}
\newcommand{\CorrMat}{\mathbf{S}}
\newcommand{\ResRatio}{\mathcal{T}}
\newcommand{\ResMat}{\mathbf{T}}

\newacronym{ap}{AP}{access point}
\newacronym{pwr}{PWR}{passive Wi-Fi radar}
\newacronym{pr}{PR}{passive radar}
\newacronym{stx}{STx}{sensing transmitter}
\newacronym{srx}{SRx}{sensing receiver}
\newacronym{rp}{RP}{radar pair}
\newacronym{ula}{ULA}{uniform linear array}
\newacronym{simo}{SIMO}{single-input and multiple-output}
\newacronym{aoa}{AoA}{angle-of-arrival}
\newacronym{doa}{DoA}{direction-of-arrival}
\newacronym{snr}{SNR}{signal-to-noise ratio}
\newacronym{awgn}{AWGN}{additive white Gaussian noise}
\newacronym{ofdm}{OFDM}{orthogonal frequency-division multiplexing}
\newacronym{ols}{OLS}{orthogonal least squares}
\newacronym{omp}{OMP}{orthogonal matching pursuit}

\newacronym{music}{MUSIC}{multiple signal classification}
\newacronym{wmusic}{WMUSIC}{weighted MUSIC}

\newacronym{imusic}{iMUSIC}{iterative MUSIC}
\newacronym{iwmusic}{iWMUSIC}{}

\newacronym{gimusic}{G-iMUSIC}{greedy iterative MUSIC}
\newacronym{giwmusic}{G-iWMUSIC}{}

\newacronym{ompimusic}{OMP-iMUSIC}{}
\newacronym{olsimusic}{OLS-iMUSIC}{}
\newacronym{ompiwmusic}{OMP-iWMUSIC}{}
\newacronym{olsiwmusic}{OLS-iWMUSIC}{}

\newacronym{esprit}{ESPRIT}{estimation of signal parameters via rotational invariance techniques}
\newacronym{ml}{ML}{maximum likelihood}
\newacronym{mle}{MLE}{maximum likelihood estimator}
\newacronym{map}{MAP}{maximum a posteriori}
\newacronym{itc}{ITC}{information theoretic criteria}
\newacronym{aic}{AIC}{Akaike information criterion}
\newacronym{aicc}{AICc}{corrected AIC}
\newacronym{bic}{BIC}{ Bayesian information criterion}

\newacronym{rmse}{RMSE}{root mean square error}
\newacronym{pdf}{PDF}{probability density function}
\newacronym{iid}{i.i.d.}{independent and identically distributed}
\newacronym{evd}{EVD}{eigen value decomposition}
\newacronym{svd}{SVD}{Singular Value Decomposition}
\newacronym{ff}{FF}{Far-Field}
\newacronym{crlb}{CRLB}{Cramér-Rao Lower Bound}
\newacronym{fft}{FFT}{Fast Fourier Transform}
\newacronym{dft}{DFT}{Discrete Fourier Transform}
\newacronym{sota}{SOTA}{state of the art}
\newacronym{flop}{FLOP}{floating-point operation}
\newacronym{bols}{BOLS}{Block orthogonal least squares}
\newacronym{bomp}{BOMP}{Block orthogonal matching pursuit}
\newacronym{mimo}{MIMO}{Multiple-input and multiple-output}
\newacronym{isac}{ISAC}{integrated sensing and communication}
\newacronym{6g}{6G}{Sixth Generation}
\newtheorem{assumption}{Assumption}
\newtheorem{lemma}{Lemma} 
\newtheorem{proposition}{Proposition} 
\newtheorem{corollary}{Corollary} 

\title{\Acrshort{gimusic}: Greedy Iterative \Acrshort{music} Algorithms \\ for Multi-Target \Acrshort{doa} Estimation}

\author{
Martin Willame,~\IEEEmembership{Graduated Student Member,~IEEE}, Gilles Monnoyer,~\IEEEmembership{Member,~IEEE},  \\
François Horlin,~\IEEEmembership{Fellow,~IEEE} and 
Jérôme Louveaux,~\IEEEmembership{Fellow,~IEEE}
\thanks{This paper was submitted on the 26th of May, 2026. Martin Willame is with the Université Catholique de Louvain (UCLouvain) and the Université Libre de Bruxelles (ULB). Gilles Monnoyer and Jérôme Louveaux are with the UCLouvain (e-mails: $\{$martin.willame, gilles.monnoyer, jerome.louveaux@uclouvain.be$\}$). François Horlin is with the ULB (e-mails: francois.horlin@ulb.be$\}$).}
}
\markboth{Submitted to IEEE Transactions on Signal Processing}%
{Willame \MakeLowercase{\textit{et al.}}: \Acrshort{gimusic}: Greedy Iterative \Acrshort{music} Algorithms \\ for Multi-Target \Acrshort{doa} Estimation}

\maketitle

\begin{abstract}
This paper presents novel algorithms for multi-target \gls{doa} estimation in array signal processing.
Although the \gls{mle} asymptotically attains the Cramér-Rao bound, its exponential complexity motivates practical alternatives, such as greedy or subspace-based methods.
In this context, greedy methods such as \gls{omp} and \gls{ols} are sensitive to early selection errors, especially for angularly proximate targets, whereas subspace-based methods such as \gls{music} present angular super-resolution capabilities but degrade under strong inter-target signal correlation.
To overcome these limitations, we propose two \gls{gimusic} algorithms, namely \acrshort{ompimusic} and \acrshort{olsimusic}, derived from a unified framework that links subspace and greedy estimations.
Unlike prior \acrshort{imusic} approaches, the proposed methods require only one initial \gls{evd} and avoid computing eigendecomposition at each iteration.
They also admit \gls{fft}-accelerated implementations for \glspl{ula}, enabling low-complexity operation.
Monte Carlo simulations demonstrate improved detection and precision over conventional \gls{omp}, \gls{ols}, and \gls{music}, as well as reduced processing time compared to greedy baselines.
Finally, we introduce diagnostic metrics that interpret performance across signal correlation and angular proximity regimes, supporting generalization beyond the specific \gls{ofdm} radar scenario considered.
\end{abstract}


\section{Introduction} \label{sec:introduction}
\glsresetall

\IEEEPARstart{M}{ulti}-target \gls{doa} estimation remains a central problem in array signal processing, with applications across the radar, sonar, wireless communications, and acoustics fields \cite{pesavento_three_2023,salama2025DoA}.
Given noisy measurements collected by a sensor array, the objective is to recover the set of source directions with high resolution and reliability.
Although this objective is conceptually simple, practical performance is fundamentally constrained by a limited array aperture, finite sample support, low \gls{snr} and inter-target correlation, while real-time operation further imposes strict computational requirements.
In this paper, we study algorithmic designs that jointly improve statistical robustness, angular resolution, and computational efficiency for multi-target \gls{doa} estimation. 

\subsection{Related Work and Motivation}
To motivate our contribution, we review state-of-the-art methods for multi-target \gls{doa} estimation and highlight their trade-offs.
Table~\ref{tab:doa_comparison} summarizes these methods qualitatively.
The \gls{mle} is asymptotically efficient under \gls{awgn} by reaching the Cramér-Rao bound as the number of observations grows \cite{928686,Foutz2008DoA,salama2025DoA}.
However, it leads to a non-convex least-squares optimization problem whose complexity scales exponentially with the number of targets, limiting its practical use.

To reduce this complexity, greedy methods such as \gls{omp} \cite{pati_orthogonal_1993,tropp_greed_2004} and \gls{ols} \cite{chen_orthogonal_1989,blumensath_difference_2007} provide efficient approximations \cite{bourguignon_sparse_2011,pesavento_three_2023}.
Greedy algorithms iteratively build the solution by selecting one candidate \gls{doa} at a time, based on a selection criterion that approximates the \gls{mle}.
While \gls{ols} generally provides more accurate approximations to \gls{mle} solutions \cite{wen_orthogonal_2021,pesavento_three_2023}, \gls{omp} remains popular in radar because of its lower cost.
When the targets present angularly proximate \gls{doa}, \gls{ols} offers clear gains, as shown both empirically \cite{mukhopadhyay_signal_2016} and theoretically \cite{soussen_joint_2013,10979421}. 
Nevertheless, greedy methods remain sensitive to early selection steps, whose capacity to resolve angularly proximate targets is constrained by the Rayleigh limit \cite{RayleighLimit}.
This can yield erroneous first estimates that propagate errors across iterations.

Subspace methods, notably \gls{music} and \gls{wmusic}, provide a complementary paradigm for \gls{doa} estimation \cite{1143830,61541,7322593,10501703}.
Such methods exploit information in the signal structure to achieve angular super-resolution, defined as the ability to resolve closely spaced targets beyond the Rayleigh limit in a single step.
Empirical studies indicate that \gls{music} often surpasses \gls{wmusic} \cite{61541}, showing the superiority of unweighted subspace methods.
However, subspace methods degrade when target signals are highly correlated \cite{pesavento_three_2023}.

Consequently, prior work has proposed \gls{imusic} variants \cite{5745290,10537353,6714005,9419990,rs14174260}, either through successive greedy detections \cite{5745290,10537353} or iterative refinement of all estimates \cite{6714005,9419990,rs14174260}.
Yet these approaches share two key limitations: they lack a principled statistical foundation, and they typically require repeated \gls{evd} at each iteration. 
The former can limit performance, while the latter is costly for real-time operation.

\begin{table*}[t]
\centering
\caption{Qualitative comparison of state-of-the-art multi-target \Acrshort{doa} estimation methods and the proposed \acrshort{imusic} variants.}
\label{tab:doa_comparison}
\renewcommand{\arraystretch}{1.2}
\setlength{\tabcolsep}{6pt}
\begin{tabular}{lllllllll}
\toprule
\textbf{Method} &
\makecell[l]{\textbf{\Acrshort{doa}}    \\ \textbf{References}} &
\textbf{Model} &
\makecell[l]{\textbf{Estimation}    \\ \textbf{Strategy}}&
\textbf{Resolution} &
\makecell[l]{\textbf{Angular}    \\ \textbf{Proximity} \\ \textbf{Handling}} &
\makecell[l]{\textbf{Signal}    \\ \textbf{Correlation} \\ \textbf{Handling}} &
\makecell[l]{\textbf{Computational}    \\ \textbf{Complexity}} &
\makecell[l]{\textbf{Key}    \\ \textbf{Limitation}} \\
\midrule
\Acrshort{mle} &
\cite{salama2025DoA,Foutz2008DoA} &
Parametric &
\makecell[l]{Multidimensional \\Search}  &
Optimal  &
Optimal &
Optimal &
Very high &
\makecell[l]{Nonlinear \\ optimization} \\
\midrule
\Acrshort{omp} &
\cite{pesavento_three_2023,10979421} &
Sparse  &
\makecell[l]{Iterative \\ greedy pursuit}  &
\makecell[l]{Angular \\ Rayleigh Limit}  &
Average &
Good &
Low &
\makecell[l]{Poor for angularly \\ proximate targets}   \\
\midrule
\Acrshort{ols} &
\cite{pesavento_three_2023,10979421} &
Sparse  &
\makecell[l]{Iterative \\ greedy pursuit}  &
\makecell[l]{Angular \\ Rayleigh Limit}  &
Good &
Very Good &
Moderate &
\makecell[l]{Poor in first steps \\ for angularly \\ proximate targets} \\
\midrule
\Acrshort{music} &
\cite{salama2025DoA,1143830} &
Subspace &
\makecell[l]{Monodimensional \\Search}  &
\makecell[l]{Super \\Resolution}  &
Very Good &
Bad &
Very Low &
\makecell[l]{Poor for \\ correlated signals}   \\
\midrule
\midrule
\Acrshort{ompimusic} &
This work &
\makecell[l]{Subspace \\ \& Sparse}  &
\makecell[l]{Iterative \\ greedy pursuit}  &
\makecell[l]{Super \\Resolution}   &
Good  &
Good &
Very Low &
\makecell[l]{Lower Performance \\ than \Acrshort{olsimusic}} \\
\midrule
\Acrshort{olsimusic} &
This work &
\makecell[l]{Subspace \\ \& Sparse}   &
\makecell[l]{Iterative \\ greedy pursuit}  &
\makecell[l]{Super \\Resolution}   &
Very Good &
Very Good &
Low &
\makecell[l]{Higher Complexity \\ than \Acrshort{ompimusic}}  \\
\bottomrule
\end{tabular}
\end{table*}

\subsection{Main Contributions}
To address these limitations, we propose two \gls{gimusic} methods---\textbf{\gls{ompimusic}} and \textbf{\gls{olsimusic}}---that combine the greedy frameworks of \gls{omp} and \gls{ols}, with the subspace-based approach of \gls{music}.
Table~\ref{tab:doa_comparison} qualitatively compares the proposed algorithms with state-of-the-art methods.
Our main contributions are:

\begin{enumerate} [wide, labelindent=0pt]

\item \textbf{Greedy iterative \acrshort{music} algorithms with improved estimation performance:}
We develop two \gls{gimusic} methods for multi-target \gls{doa} estimation, along with their weighted \gls{imusic} variants.
Unlike existing iterative \gls{music} heuristics, these methods rely on a unified framework that links subspace estimation and greedy optimization.
The proposed algorithms achieve better localization performance than the subspace-based \gls{music} methods and their greedy \gls{omp} and \gls{ols} counterparts while exhibiting lower computational complexity.

\item \textbf{Single-\acrshort{evd} design with \Acrshort{fft}-accelerated implementation:}
A central distinction from prior \gls{imusic} methods is that, in the proposed unified framework, \gls{gimusic} algorithms are shown to avoid repeated \gls{evd}. 
Instead, they require only one initial \gls{evd} to estimate the signal and noise subspaces, and then iteratively update them using efficient updates.
In addition, we introduce \gls{fft}-accelerated implementations in \gls{ula} settings for \gls{ols} and the proposed \gls{gimusic} variants, extending such acceleration beyond existing \gls{music} and \gls{omp} implementations.
As a result, these design choices substantially reduce complexity and support real-time radar operation under strict latency constraints.

\item \textbf{Comprehensive complexity and performance characterization:}
We provide a detailed theoretical complexity analysis together with extensive Monte Carlo simulations benchmarked against state-of-the-art methods.
This study characterizes operational regimes as functions of key radar parameters and identifies where \gls{gimusic} methods provide the largest gains.

\item \textbf{Practical passive \acrshort{ofdm} radar validation with diagnostic metrics:}
We validate the proposed \gls{gimusic} algorithms in a practical passive \gls{ofdm} radar scenario without decoding unknown data symbols.
Beyond this use case, the methods can be applied to other sensing scenarios.
To support this extension, we introduce two diagnostic metrics: the \textit{steering-vector correlation metric} ($\ResRatio$) and the \textit{signal correlation metric} ($\CorrRatio$), which quantify the mean inter-target angular proximity and signal correlation, respectively.
We then interpret performance through these metrics to provide scenario-agnostic guidance.
\end{enumerate}

\subsection{Paper Organization}
The remainder of this paper is organized as follows.
Section~\ref{sec:system_model} introduces the system model. 
Section~\ref{sec:state_of_the_art} reviews state-of-the-art algorithms for multi-target \gls{doa} estimation. 
Section~\ref{sec:imusic} presents the proposed \gls{gimusic} algorithms.
Section~\ref{sec:fft_acceleration} details \gls{fft} accelerations in \gls{ula} settings. 
Section~\ref{sec:Complexity Analysis} provides a comprehensive complexity analysis. 
Section~\ref{sec:simulations} demonstrates performance gains through extensive Monte Carlo simulations. 
Finally, Section~\ref{sec:conclusion} concludes the paper.

\subsection{Mathematical Notation}
Scalars, vectors, and matrices are denoted by $a$, $\mathbf{a}$, and $\mathbf{A}$, respectively.
In the estimation context, $\theta$, $\CandVal{\theta}$, and $\EstVal{\theta}$ denote the true, candidate, and estimated values of a parameter.
The sets of real and complex numbers are denoted by $\Rset$ and $\Cset$.
The trace, transpose, Hermitian transpose, and Moore-Penrose pseudoinverse operators are denoted by $\Tr\left\{\mathbf{A}\right\}$, $\Transpose{\mathbf{A}}$, $\Hermitian{\mathbf{A}}$, and $\MPInverse{\mathbf{A}}$, respectively.
The vectorization operator that stacks all columns of a matrix is denoted by $\Vectorize{\mathbf{A}}$.
For a positive semidefinite matrix $\mathbf{A}$, its square-root matrix is denoted by $\SqrtMat{\mathbf{A}}$ such that $\SqrtMat{\mathbf{A}}\,\Hermitian{(\SqrtMat{\mathbf{A}})}=\mathbf{A}$.
The identity matrix is denoted by $\mathbf{I}_\SRxAll \in \Rset^{\SRxAll \times \SRxAll}$.
The vector $\ell_2$ norm and the matrix Frobenius norm are denoted by $\lVert \mathbf{a} \rVert_2$ and $\lVert \mathbf{A} \rVert_\Frob$, while $\lVert \mathbf{A} \rVert_{2,\Col}$ denotes the column-wise $\ell_2$ norm.

\section{System Model} \label{sec:system_model}
In order to ground the proposed algorithms in a concrete application scenario, we consider a \gls{pr} system designed to estimate the \gls{doa} of multiple targets using the reflected signals transmitted by an \gls{ap}. 
The \gls{ap} and the \gls{pr} can be either co-located (monostatic configuration) or spatially separated (bistatic configuration).
The single-antenna \gls{ap} transmits an \gls{ofdm} downlink signal that the passive receiver captures using a \gls{ula} of $\SRxAll$ antennas with an antenna spacing of $\ArraySpacing$.
The \gls{ap} transmits a sequence of $\ObsAll$ \gls{ofdm} data symbols, separated by $\TimeFrame$ seconds.
Each \gls{ofdm} symbol contains $\SubAll$ subcarriers with uniform spacing $\SubSpacing$.
The first subcarrier has frequency $f_c$ and corresponding wavenumber $\WaveNumber \triangleq 2\pi f_c / c$, where $c$ denotes the speed of light. 
The transmitted complex data symbol for the \ObsTh \gls{ofdm} symbol ($\ObsInd \in \{0,\dots,\ObsAll-1\}$) and \SubTh subcarrier ($\SubInd \in \{0,\dots,\SubAll-1\}$) is denoted as $\SymbolVal_{\ObsInd,\SubInd} \in \Cset$.
These symbols remain unknown to the passive receiver.

The \gls{pr} objective is to estimate the \gls{doa} of $\TgtAll$ targets within the coverage area.
We define the true \gls{doa} of the \TgtTh target as $\Aoa_\TgtInd \in [-\pi/2,\pi/2]$. 
Collecting all target \glspl{doa}, we define the \gls{doa} vector $\AoaVec=\Transpose{[\Aoa_1 \ \dots \ \Aoa_\TgtAll]} \in \Rset^{\TgtAll \times 1}$.
\glspl{doa} measure the angle between the incoming wavefront and the normal vector of the receiver antenna array. 
Under the far-field assumption for all targets with respect to the receiver, the \gls{doa} steering vector $\SteerVec(\Aoa) \in \Cset^{\SRxAll \times 1}$ and the corresponding steering matrix $\SteerMat(\AoaVec) \in \Cset^{\SRxAll \times \TgtAll}$ are defined as
\begin{align}
    \SteerVec(\Aoa) & = \Transpose{\big[1 \ e^{j \WaveNumber\ArraySpacing \sin(\Aoa)} \dots \ e^{j \WaveNumber\ArraySpacing \sin(\Aoa) (\SRxAll-1) }\big]}, \\
    \SteerMat(\AoaVec) & = \big[\SteerVec(\Aoa_1) \ \dots \ \SteerVec(\Aoa_\TgtAll)\big].
\end{align}
We ground the channel model on the following assumptions:
first, only single-bounce multipath signals contribute significantly to the observed channel; second, the receiver achieves perfect timing and frequency synchronization using the transmitted preamble.
Under these assumptions, the channel for the \ObsTh \gls{ofdm} symbol and \SubTh subcarrier is expressed as $\ChnlVec_{\ObsInd,\SubInd}(\AoaVec,\CoefVec_{\ObsInd,\SubInd}) = \SteerMat(\AoaVec) \ \CoefVec_{\ObsInd,\SubInd}$, with, 
\begin{align}
    \CoefVec_{\ObsInd,\SubInd} & = \Transpose{[\CoefVal_{\ObsInd,\SubInd,1} \ \dots \ \CoefVal_{\ObsInd,\SubInd,\TgtAll}]}, \in \Cset^{\TgtAll \times 1}, \\
    \CoefVal_{\ObsInd,\SubInd,\TgtInd} & = \PathVal_{\TgtInd} \ e^{-j2\pi \WaveNumber \Range_{\TgtInd}} \ e^{-j2\pi \SubSpacing \Range_{\TgtInd} \SubInd} \ e^{j2\pi \Doppler_{\TgtInd} \ObsInd/\ObsAll \TimeFrame}, \in \Cset.
\end{align}
where the complex channel coefficient vector $\CoefVec_{\ObsInd,\SubInd}$ captures, for each target $\TgtInd$, the path loss, the radar cross-section and the propagation phase (all encoded in $\alpha_{\TgtInd} \in \Cset$), the range-induced phase shift ($\Range_{\TgtInd}$), and the Doppler shift ($\Doppler_{\TgtInd}$) \cite{9477585}.  
The model characterizing the received signal $\SRxSigVec_{\ObsInd,\SubInd} \in \Cset^{\SRxAll \times 1}$ at the \gls{pr} is given by 
\begin{equation} \label{eq:SRx_sig}
\SRxSigVec_{\ObsInd,\SubInd} = \ChnlVec_{\ObsInd,\SubInd}(\AoaVec,\CoefVec_{\ObsInd,\SubInd}) \ \SymbolVal_{\ObsInd,\SubInd} + \NoiseVec_{\ObsInd,\SubInd} = \SteerMat(\AoaVec) \ \CoefVec_{\ObsInd,\SubInd}^{\prime} + \NoiseVec_{\ObsInd,\SubInd},
\end{equation}
where $\CoefVec_{\ObsInd,\SubInd}^{\prime}\triangleq \CoefVec_{\ObsInd,\SubInd} \ \SymbolVal_{\ObsInd,\SubInd} \in \Cset^{\TgtAll \times 1}$ represents the channel coefficients modulated by the unknown data symbol, and $\NoiseVec_{\ObsInd,\SubInd} \in \Cset^{\SRxAll \times 1}$ denotes the \gls{awgn} contribution.
Each noise vector $\NoiseVec_{\ObsInd,\SubInd} \overset{\text{\acrshort{iid}}}{\sim} \CN(\mathbf{0}, \NoiseVar \ \Identity_\SRxAll)$ is an \gls{iid} circularly symmetric complex Gaussian random vector with zero mean and diagonal covariance matrix. 
The corresponding \gls{snr} is defined as
\begin{equation} \label{eq:snr}
    \SNR = \frac{\frac{1}{\TgtAll}\sum_{\TgtInd=1}^{\TgtAll} |\PathVal_\TgtInd |^2}{\NoiseVar}. 
\end{equation}
With this definition, the \gls{snr} is dominated by the strongest target, which is also the "easiest" to detect. 
The received signal matrix $\SRxSigMat = [\SRxSigVec_{0,0} \ \dots \ \SRxSigVec_{\ObsAll-1,\SubAll-1}]$ is expressed by juxtaposing all observation vectors as
\begin{equation} \label{eq:SRx_mat}
\SRxSigMat = \SteerMat(\AoaVec) \ \CoefMat^{\prime} + \NoiseMat, \in \Cset^{\SRxAll \times \ObsAll \SubAll},
\end{equation}
where the modulated channel coefficient and noise matrices are defined as
\begin{align}
\CoefMat^{\prime} & = [\CoefVec_{0,0}^{\prime} \ \dots \ \CoefVec_{\ObsAll-1,\SubAll-1}^{\prime}], \in \Cset^{\TgtAll \times \ObsAll \SubAll}, \\
\NoiseMat & = [\NoiseVec_{0,0} \ \dots \ \NoiseVec_{\ObsAll-1,\SubAll-1}], \in \Cset^{\SRxAll \times \ObsAll \SubAll}. 
\end{align}
\section{State-of-the-Art Methods for Multi-Target \Acrshort{doa} Estimation} \label{sec:state_of_the_art}
This section reviews the main state-of-the-art \gls{doa} estimators and provides the technical foundation for the proposed \gls{gimusic} approaches introduced in Section~\ref{sec:imusic}.
Throughout this section, we assume that the number of targets $\TgtAll$ is known. 
In practice, $\TgtAll$ can efficiently be estimated with model-order selection methods.
The numerical results in Section~\ref{sec:simulations} integrate and compare two variants of such strategies.

\subsection{Maximum-Likelihood Estimation} \label{sec:MLE}
We first formulate the \gls{mle} for multi-target \gls{doa} estimation.
Assuming known (or pre-estimated) noise variance $\NoiseVar$, we gather the unknown parameters in
$\boldsymbol \gamma = \Transpose{[\Transpose{\AoaVec} \ \Transpose{\Vectorize{\CoefMat^{\prime}}}]}$.
The \gls{mle} identifies the estimate $\widehat{\boldsymbol \gamma}$ that maximizes the likelihood of $\SRxSigMat$ under the observation model in \eqref{eq:SRx_mat}.
Because observations are independent across \gls{ofdm} symbols and subcarriers, the likelihood factorizes into Gaussian terms; after taking the logarithm and removing constants, we obtain
\begin{equation} \label{eq:global_max}
\widehat{\boldsymbol \gamma} 
= \argmin_{\widetilde{\boldsymbol \gamma}} \left \lVert \SRxSigMat - \SteerMat(\widetilde{\AoaVec}) \ \widetilde{\CoefMat}^{\prime} \right \rVert_\Frob^2.
\end{equation}
We next solve \eqref{eq:global_max} with respect to the channel coefficients $\CandVal{\CoefMat}^{\prime}$, which gives a closed-form expression as a function of $\CandVal{\AoaVec}$:
\begin{equation} \label{eq:Moore-Penrose}
\EstVal{\CoefMat}^{\prime} = \Inverse{\left(\Hermitian{\AoaSteerMat}(\CandVal{\AoaVec}) \AoaSteerMat(\CandVal{\AoaVec})\right)} \Hermitian{\AoaSteerMat}(\CandVal{\AoaVec}) \ \SRxSigMat = \MPInverse{\AoaSteerMat}(\CandVal{\AoaVec}) \ \SRxSigMat,
\end{equation}
where $\MPInverse{\AoaSteerMat}(\CandVal{\AoaVec}) = \Inverse{\big(\Hermitian{\AoaSteerMat}(\CandVal{\AoaVec}) \AoaSteerMat(\CandVal{\AoaVec})\big)} \Hermitian{\AoaSteerMat}(\CandVal{\AoaVec})$ denotes the Moore-Penrose pseudoinverse \cite{salama2025DoA}. 
By substituting \eqref{eq:Moore-Penrose} into \eqref{eq:global_max} and introducing the projection matrix onto the steering-matrix column space,
$\ProjMat(\CandVal{\AoaVec}) = \AoaSteerMat(\CandVal{\AoaVec}) \ \MPInverse{\AoaSteerMat}(\CandVal{\AoaVec})$, we obtain the concentrated log-likelihood problem:
\begin{equation} \label{eq:concentrated_MLE}
\widehat{\AoaVec} = \argmin_{\CandVal{\AoaVec}} \left \lVert \SRxSigMat - \ProjMat(\CandVal{\AoaVec}) \ \SRxSigMat \right \rVert_\Frob^2.
\end{equation}
After algebraic manipulations, this problem reduces to
\begin{align*}
\widehat{\AoaVec} 
& = \argmax_{\CandVal{\AoaVec}} \left \lVert \Hermitian{\SRxSigMat} \ \ProjMat(\CandVal{\AoaVec}) \right \rVert_\Frob^2, \numberthis \label{eq:optimization_problem_Norm} \\
& = \argmax_{\CandVal{\AoaVec}}  \Tr\left[ \SCovMat \ \mathbf{P}(\CandVal{\AoaVec}) \right], \numberthis \label{eq:optimization_problem_Trace} \\
& = \argmax_{\CandVal{\AoaVec}} \left \lVert \Hermitian{(\SqrtMat{\SCovMat})} \ \ProjMat(\CandVal{\AoaVec}) \right \rVert_\Frob^2. \numberthis \label{eq:optimization_problem_Rsqrt}
\end{align*}
where we define $\SCovMat = \frac{1}{\ObsAll \SubAll} \SRxSigMat \Hermitian{\SRxSigMat} \in \Cset^{\SRxAll \times \SRxAll}$ as the sample covariance matrix of the received signal.
All three forms in \eqref{eq:optimization_problem_Norm}, \eqref{eq:optimization_problem_Trace}, and \eqref{eq:optimization_problem_Rsqrt} are equivalent, but the latter is less conventional in the literature.
It is presented because it is useful for the developments in Sections~\ref{sec:OMP_OLS} and \ref{sec:imusic}.
It is expressed in terms of the square-root residual sample covariance matrix $\SqrtMat{\SCovMat}$ that can have any size for the column dimension, as long as $\SqrtMat{\SCovMat} \ \Hermitian{\SqrtMat{\SCovMat}} = \SCovMat$ holds.
As such, $\SRxSigMat, \in \Cset^{\SRxAll \times \ObsAll \SubAll}$ is a valid square-root sample covariance matrix. 
In this work, we use $\SqrtMat{\SCovMat} \in \Cset^{\SRxAll \times \SRxAll}$, which is more compact and computationally efficient than $\SRxSigMat$ as shown in the complexity analysis of Section~\ref{sec:Complexity Analysis}.
Such $\SqrtMat{\SCovMat}$ form can be obtained via Cholesky decomposition, \gls{evd}, or other methods. 

The \gls{mle} solution is given by the $\TgtAll$ \gls{doa} estimates in $\CandVal{\AoaVec}$ that maximize \eqref{eq:optimization_problem_Norm} (or equivalently \eqref{eq:optimization_problem_Trace} or \eqref{eq:optimization_problem_Rsqrt}). 
In practice, however, exact \gls{mle} remains computationally prohibitive, with complexity scaling as $\bigO(\GridAll^{\TgtAll}\SRxAll^3)$, where $\GridAll$ is the number of grid points per \gls{doa} variable.

\subsection{Subspace Methods: \Acrshort{music} and \Acrshort{wmusic}} \label{sec:MUSIC}
Because exact \gls{mle} is intractable in many scenarios, subspace methods such as \gls{music} and \gls{wmusic} are widely used alternatives.
They provide super-resolution behavior at substantially lower complexity \cite{1143830,61541,7322593,10501703}.
Empirical evidence further indicates that \gls{music} often outperforms its weighted counterpart in peak detection \cite{61541}.
Algorithm~\ref{alg:music} summarizes this workflow, which is detailed below.
\begin{algorithm}[t]
    \caption{\gls{music} and \gls{wmusic}}
    \label{alg:music}
    \KwIn{$\TgtAll$ and $\SCovMat$}
    \vspace{0.05cm}
    \KwOut{$\widehat{\AoaVec} = [\widehat{\Aoa}_1 \ \dots  \ \widehat{\Aoa}_{\EstVal{\TgtAll}}]$}

    \Begin{
    \textbf{1. \Acrshort{evd}:} Compute eigenvalues/eigenvectors of $\SCovMat$\;
    \textbf{2. Signal/Noise partition:} $\SCovMat = \SigSpaceMat \SigEigMat \Hermitian{\SigSpaceMat} + \NoiseSpaceMat \NoiseEigMat \Hermitian{\NoiseSpaceMat}$\;
    \textbf{3. Pseudospectrum evaluation:} $\forall \widetilde{\Aoa} \in \Grid_\Aoa$,
    \newline \vspace{2pt} a) \gls{music}: $\PseudoSpec_{\NoiseSpaceMat}(\widetilde{\Aoa}) $ in \eqref{eq:pseudo_spectrum_noise} or $\PseudoSpec_{\SigSpaceMat}(\widetilde{\Aoa})$ in \eqref{eq:pseudo_spectrum_signal};
    \newline \vspace{2pt} b) \gls{wmusic}:$\PseudoSpec_{\NoiseEigMat\NoiseSpaceMat}(\widetilde{\Aoa})$ in \eqref{eq:pseudo_spectrum_noise_weighted} or $\PseudoSpec_{\SigEigMat\SigSpaceMat}(\widetilde{\Aoa})$ in \eqref{eq:pseudo_spectrum_signal_weighted}\; 
    \textbf{4. Peak selection:} $\widehat{\AoaVec} \gets \TgtAll$ largest peaks in $\PseudoSpec$\;
    }
\end{algorithm}

\subsubsection{Eigenvalue Decomposition}
The first step computes the \gls{evd} of the sample covariance matrix $\SCovMat$.

\subsubsection{Signal- and Noise-Subspace Partition} \label{sec:music_sig_noise_partition}
We partition $\SCovMat$ into signal and noise components using the computed eigenpairs.
The dimensionality of the signal subspace is typically set to match the number $\TgtAll$ of targets.
The signal subspace $\SigSpaceMat \in \Cset^{\SRxAll \times \TgtAll}$ spans the $\TgtAll$ dominant eigenvectors, while the noise subspace $\NoiseSpaceMat \in \Cset^{\SRxAll \times (\SRxAll - \TgtAll)}$ spans the remaining eigenvectors.
Hence,
\begin{equation} \label{eq:SVD_SCovMat}
    \SCovMat = \SigSpaceMat \SigEigMat \Hermitian{\SigSpaceMat} + \NoiseSpaceMat \NoiseEigMat \Hermitian{\NoiseSpaceMat},
\end{equation}
where the diagonal signal and noise eigenvalue matrices are denoted as $\SigEigMat$ and $\NoiseEigMat$, respectively.

\subsubsection{Pseudospectrum Evaluation}
The \gls{music} pseudospectrum is evaluated over a search grid $\Grid_\Aoa$ of $\GridAll$ candidate \gls{doa} values.
To reduce complexity, we evaluate either a noise-subspace or a signal-subspace form, depending on the subspace dimensions.
When the signal subspace is larger ($\TgtAll > \SRxAll - \TgtAll$), we use $\NoiseSpaceMat \Hermitian{\NoiseSpaceMat}$; otherwise, we use $\SigSpaceMat \Hermitian{\SigSpaceMat}$.
The corresponding pseudo-spectra are
\begin{align} 
    \PseudoSpec_{\NoiseSpaceMat}(\CandVal{\Aoa}) 
    & = \frac{1}{\Hermitian{\SteerVec}(\CandVal{\Aoa}) \ \NoiseSpaceMat \Hermitian{\NoiseSpaceMat} \ \SteerVec(\CandVal{\Aoa})}
    = \left\lVert \Hermitian{\NoiseSpaceMat} \ \SteerVec(\CandVal{\Aoa}) \right\rVert_{2}^{-2}, \label{eq:pseudo_spectrum_noise} \\
    \PseudoSpec_{\SigSpaceMat}(\CandVal{\Aoa}) 
    & = \Hermitian{\SteerVec}(\CandVal{\Aoa}) \ \SigSpaceMat \Hermitian{\SigSpaceMat} \ \SteerVec(\CandVal{\Aoa}) 
    = \left\lVert \Hermitian{\SigSpaceMat} \ \SteerVec(\CandVal{\Aoa}) \right\rVert_{2}^{2}. \label{eq:pseudo_spectrum_signal}
\end{align}
Although $\PseudoSpec_{\NoiseSpaceMat}$ and $\PseudoSpec_{\SigSpaceMat}$ differ numerically, they share identical peak locations because
$\SigSpaceMat \Hermitian{\SigSpaceMat} + \NoiseSpaceMat \Hermitian{\NoiseSpaceMat} = \Identity_\SRxAll$ \cite{61541}.
For \gls{wmusic}, the pseudo-spectrum is defined similarly but with eigenvalue weighting:
\begin{align} 
    \PseudoSpec_{\NoiseEigMat\NoiseSpaceMat}(\CandVal{\Aoa}) 
    & = \frac{1}{\Hermitian{\SteerVec}(\CandVal{\Aoa}) \NoiseSpaceMat \NoiseEigMat \Hermitian{\NoiseSpaceMat} \SteerVec(\CandVal{\Aoa})}
    = \left\lVert \Hermitian{(\SqrtMat{\NoiseEigMat}\NoiseSpaceMat)} \ \SteerVec(\CandVal{\Aoa}) \right\rVert_{2}^{-2}, \label{eq:pseudo_spectrum_noise_weighted} \\
    \PseudoSpec_{\SigEigMat\SigSpaceMat}(\CandVal{\Aoa}) 
    & = \Hermitian{\SteerVec}(\CandVal{\Aoa}) \ \SigSpaceMat \SigEigMat \Hermitian{\SigSpaceMat} \ \SteerVec(\CandVal{\Aoa}) 
    = \left\lVert \Hermitian{(\SqrtMat{\SigEigMat}\SigSpaceMat)} \ \SteerVec(\CandVal{\Aoa}) \right\rVert_{2}^{2}. \label{eq:pseudo_spectrum_signal_weighted}
\end{align}
Since ${\SigEigMat}$ and ${\NoiseEigMat}$ are diagonal, their square-root matrices are obtained by taking the square root of each eigenvalue.
By contrast to standard \gls{music}, the weighted forms $\PseudoSpec_{\NoiseEigMat\NoiseSpaceMat}$ and $\PseudoSpec_{\SigEigMat\SigSpaceMat}$ do not necessarily peak at the same locations, because weighting breaks the complementary projection relation in \eqref{eq:SVD_SCovMat}.
As a result, these two weighted variants may perform differently depending on the scenario.

\subsubsection{Peak Selection}
We then identify the $\TgtAll$ most prominent peaks of the pseudo-spectrum to obtain
$\widehat{\AoaVec} = [\widehat{\Aoa}_1 \ \dots \ \widehat{\Aoa}_{\TgtAll}]$.
Optionally, gradient-ascent refinement at each detected peak mitigates grid quantization effects and enables off-grid localization in the continuous domain.


\subsection{Greedy Methods: \Acrshort{omp} and \Acrshort{ols}} \label{sec:OMP_OLS}
We now move from subspace estimation to greedy methods that approximate \gls{mle} at lower cost.
Both \gls{omp} and \gls{ols} estimate targets by selecting one \gls{doa} per iteration while leveraging previously selected \glspl{doa}.
We denote the estimated \gls{doa} set at iteration $k$ as $\GreedySet{\EstVal{\AoaVec}}{k} = [\EstVal{\Aoa}_1 \ \dots \ \EstVal{\Aoa}_k]$.
The projection operators onto the subspace generated by the columns of $\SteerMat(\GreedySet{\EstVal{\AoaVec}}{k})$, and onto its orthogonal complement, are defined as
\begin{align}
& \ProjMat(\GreedySet{\EstVal{\AoaVec}}{k}) = \AoaSteerMat(\GreedySet{\EstVal{\AoaVec}}{k}) \ \MPInverse{\AoaSteerMat}(\GreedySet{\EstVal{\AoaVec}}{k}), \\
& \CompProjMat(\GreedySet{\EstVal{\AoaVec}}{k}) = \Identity_\SRxAll - \ProjMat(\GreedySet{\EstVal{\AoaVec}}{k}). \label{eq:simo_ProjCompMat}
\end{align}
Algorithm~\ref{alg:greedy} summarizes the common framework of \gls{omp} and \gls{ols}, which we detail next.
\begin{algorithm}[t]
    \caption{\gls{omp} and \gls{ols}}
    \label{alg:greedy}
    \KwIn{$\TgtAll$ and $\SRxSigMat$ or $\SCovMat$}
    \vspace{0.05cm}
    \KwOut{$\widehat{\AoaVec}_{\TgtAll} = [\widehat{\Aoa}_1 \ \dots  \ \widehat{\Aoa}_{\TgtAll}]$}

    \Begin{
    \textbf{1. Initialization:} $\GreedySet{\EstVal{\AoaVec}}{0} \gets [~]$, $\CompProjMat \big(\GreedySet{\EstVal{\AoaVec}}{0}\big) \gets \Identity_\SRxAll$\;
        \For{$k\gets 0$ \KwTo $\TgtAll-1$}{
        	\textbf{2. Selection Step:} \newline \vspace{2pt} a) \gls{omp}: $\EstVal{\Aoa}_{\TgtInd+1} \gets $ \eqref{eq:simo_OMP_Y} or \eqref{eq:simo_OMP_R}; \newline \vspace{2pt}  b) \gls{ols}: $~\EstVal{\Aoa}_{\TgtInd+1} \gets $ \eqref{eq:simo_OLS_Y} or \eqref{eq:simo_OLS_R}\;
        	\textbf{3. Update Step:} $ \GreedySet{\EstVal{\AoaVec}}{\TgtInd+1} \gets [\GreedySet{\EstVal{\AoaVec}}{\TgtInd},\EstVal{\Aoa}_{\TgtInd+1}], \CompProjMat\big(\GreedySet{\EstVal{\AoaVec}}{\TgtInd+1}\big) \gets $ \eqref{eq:simo_ProjCompMat}\; 
    	}
    }
\end{algorithm}

\subsubsection{Initialization}
Both algorithms initialize with an empty \gls{doa} set $\GreedySet{\EstVal{\AoaVec}}{0}$ and target counter $k=0$.
At initialization, $\CompProjMat(\GreedySet{\EstVal{\AoaVec}}{0}) = \Identity_\SRxAll$, which spans the full observation space.
The algorithm then alternates between selection and update steps until $\TgtAll$ targets are estimated.

\subsubsection{Selection Step}
To express the selection criterion, we define the residual observation and covariance matrices as
\begin{equation}
\SRxSigMat_{\TgtInd} = \ProjMat^{\Comp}(\GreedySet{\EstVal{\AoaVec}}{k}) \ \SRxSigMat \quad \text{and} \quad 
\GreedySet{\SCovMat}{\TgtInd} = \ProjMat^{\Comp}(\GreedySet{\EstVal{\AoaVec}}{k}) \ \SCovMat \ \ProjMat^{\Comp}(\GreedySet{\EstVal{\AoaVec}}{k}).
\end{equation}
Each iteration selects the next target \gls{doa} by solving a one-dimensional optimization problem over $\Grid_\Aoa$.
Similarly to \gls{music}, optional gradient-ascent refinement can be used for off-grid estimation.
The \gls{omp} and \gls{ols} methods differ in the objective function, as detailed below. 
\gls{omp} selects the next target \gls{doa} $\EstVal{\Aoa}_{\TgtInd+1}$ by maximizing the correlation between the steering vector and the residual observation matrix
\begin{align}
\EstVal{\Aoa}_{\TgtInd+1}^{\text{\Acrshort{omp}}} 
    & = \argmax_{\CandVal{\Aoa}_{\TgtInd+1}} \left\lVert \Hermitian{\SRxSigMat}_{\TgtInd} \ \SteerVec(\CandVal{\Aoa}_{\TgtInd+1})  \right\rVert_2^2, \numberthis \label{eq:simo_OMP_Y} \\
    & = \argmax_{\CandVal{\Aoa}_{\TgtInd+1}} \Hermitian{\SteerVec}(\CandVal{\Aoa}_{\TgtInd+1}) \ \GreedySet{\SCovMat}{\TgtInd} \ \SteerVec(\CandVal{\Aoa}_{\TgtInd+1}), \numberthis \label{eq:simo_OMP_R} \\
    & = \argmax_{\CandVal{\Aoa}_{\TgtInd+1}} \left\lVert \Hermitian{(\SqrtMat{\SCovMat}_{\TgtInd})} \ \SteerVec(\CandVal{\Aoa}_{\TgtInd+1})  \right\rVert_2^2. \numberthis \label{eq:simo_OMP_Rsqrt}
\end{align}
Again, all three forms are equivalent, but the last one is less conventional in the literature. 
It is expressed in terms of the square-root residual sample covariance matrix, defined as
\begin{equation} \label{eq:SqrtResMat}
\SqrtMat{\SCovMat}_{\TgtInd}=\CompProjMat(\GreedySet{\EstVal{\AoaVec}}{k})\ \SqrtMat{\SCovMat}, \in \Cset^{\SRxAll \times \SRxAll}. 
\end{equation}
To avoid redundant computations, we compute $\SqrtMat{\SCovMat}$ once before the iterations and then reuse it for every update of $\SqrtMat{\SCovMat}_{\TgtInd}$ from \eqref{eq:SqrtResMat}.
The \gls{omp} selection problem is equivalent to solving a univariate \gls{mle} problem on the residual observation matrix $\SRxSigMat_{\TgtInd}$ or sample covariance matrix $\SCovMat_{\TgtInd}$.
By contrast, \gls{ols} selects $\EstVal{\Aoa}_{\TgtInd+1}$ by minimizing the least-squares error at each iteration:
\begin{align}
    \EstVal{\Aoa}_{\TgtInd+1}^{\text{\Acrshort{ols}}} 
    & = \argmax_{\CandVal{\Aoa}_{\TgtInd+1}} \left \lVert \Hermitian{\SRxSigMat} \ \ProjMat\big(\big[\GreedySet{\EstVal{\AoaVec}}{\TgtInd},\CandVal{\Aoa}_{\TgtInd+1} \big]\big) \right \rVert_\Frob^2, \numberthis \label{eq:simo_OLS_Y_slow} \\
    & = \argmax_{\CandVal{\Aoa}_{\TgtInd+1}} \Tr\left[ \SCovMat \ \ProjMat\big(\big[ \GreedySet{\EstVal{\AoaVec}}{\TgtInd},\CandVal{\Aoa}_{\TgtInd+1} \big]\big)  \right], \numberthis \label{eq:simo_OLS_R_slow} \\
    & = \argmax_{\CandVal{\Aoa}_{\TgtInd+1}} \left\lVert \Hermitian{(\SqrtMat{\SCovMat}_{\TgtInd})} \ \ProjMat\big(\big[\GreedySet{\EstVal{\AoaVec}}{\TgtInd},\CandVal{\Aoa}_{\TgtInd+1} \big]\big)   \right\rVert_\Frob^2. \numberthis \label{eq:simo_OLS_Rsqrt_slow}
\end{align}
These formulations are equivalent to a $(k+1)$-dimensional \gls{mle} constrained by the previously selected set $\GreedySet{\EstVal{\AoaVec}}{\TgtInd}$, so that only the candidate \gls{doa} $\EstVal{\Aoa}_{\TgtInd+1}$ is searched over one dimension.
Intuitively, \gls{ols} jointly reoptimizes the candidate \gls{doa} and the modulated channel coefficients of all previously selected \glspl{doa} plus the candidate, whereas \gls{omp} optimizes only the candidate \gls{doa} and its corresponding coefficient.
This joint optimization improves estimation accuracy but drastically increases computationalcomplexity.
To alleviate this burden, previous works \cite{7543,10979421} have leveraged a projection-matrix decomposition demonstrating that
\begin{equation} \label{eq:proj_mat_decomposition}
\ProjMat\big(\big[\GreedySet{\EstVal{\AoaVec}}{\TgtInd},\CandVal{\Aoa}_{\TgtInd+1} \big]\big)
= \ProjMat\big(\GreedySet{\EstVal{\AoaVec}}{\TgtInd}\big) + \frac{\SteerVec_\TgtInd(\CandVal{\Aoa}_{\TgtInd+1}) \ \Hermitian{\SteerVec}_\TgtInd(\CandVal{\Aoa}_{\TgtInd+1})}{\lVert \SteerVec_\TgtInd(\CandVal{\Aoa}_{\TgtInd+1}) \rVert_2^2}, 
\end{equation} 
where $\SteerVec_\TgtInd(\CandVal{\Aoa}_{\TgtInd+1}) = \CompProjMat\big(\GreedySet{\EstVal{\AoaVec}}{\TgtInd}\big) \ \SteerVec(\CandVal{\Aoa}_{\TgtInd+1})$.
By substituting \eqref{eq:proj_mat_decomposition} into \eqref{eq:simo_OLS_Y_slow}, \eqref{eq:simo_OLS_R_slow}, and \eqref{eq:simo_OLS_Rsqrt_slow}, and then removing terms independent of $\CandVal{\Aoa}_{\TgtInd+1}$, we obtain the following efficient forms
\begin{align}
\EstVal{\Aoa}_{\TgtInd+1}^{\text{\Acrshort{ols}}}
    & = \argmax_{\CandVal{\Aoa}_{\TgtInd+1}} \frac{\lVert \Hermitian{\SRxSigMat}_{\TgtInd} \ \SteerVec(\CandVal{\Aoa}_{\TgtInd+1})  \rVert_2^2}{\lVert \SteerVec_\TgtInd(\CandVal{\Aoa}_{\TgtInd+1}) \rVert_2^2}, \label{eq:simo_OLS_Y} \\
    & = \argmax_{\CandVal{\Aoa}_{\TgtInd+1}} \ 
    \frac{\Hermitian{\SteerVec}(\CandVal{\Aoa}_{\TgtInd+1}) \ \GreedySet{\SCovMat}{\TgtInd} \ \SteerVec(\CandVal{\Aoa}_{\TgtInd+1})}{\lVert \SteerVec_\TgtInd(\CandVal{\Aoa}_{\TgtInd+1}) \rVert_2^2}, \label{eq:simo_OLS_R} \\
    & = \argmax_{\CandVal{\Aoa}_{\TgtInd+1}} \ 
    \frac{\lVert \Hermitian{(\SqrtMat{\SCovMat}_{\TgtInd})} \ \SteerVec(\CandVal{\Aoa}_{\TgtInd+1})  \rVert_2^2}{\lVert \SteerVec_\TgtInd(\CandVal{\Aoa}_{\TgtInd+1}) \rVert_2^2}. \label{eq:simo_OLS_Rsqrt}
\end{align}
The denominator captures the projection of the steering vector onto the orthogonal complement of previously selected angles, thereby normalizing the \gls{omp}-like numerator.

\subsubsection{Update Step}
After each selection, we update the angle set and the orthogonal-complement projection matrix to include the newly estimated angle $\EstVal{\Aoa}_{\TgtInd+1}$.
This update defines the residual observation space used in the next iteration.

\section{\acrlong{gimusic} Algorithms} \label{sec:imusic}
This section introduces the proposed \gls{gimusic} algorithms for multi-target \gls{doa} estimation.
These algorithms integrate \gls{music}-based subspace pseudo-spectrum evaluations into the greedy \gls{omp} and\gls{ols} frameworks to improve and accelerate target-\gls{doa} selection.
As in Section~\ref{sec:state_of_the_art}, we assume that the number of targets $\TgtAll$ is known; the discussion about practical model-order selection is deferred to Section~\ref{sec:simulations}.

We first present a unified framework that links subspace estimation and greedy optimization, thereby motivating the proposed algorithms.
We then detail the processing steps of \gls{ompimusic} and \gls{olsimusic}, along with their weighted \gls{imusic} variants.

\subsection{Unified Subspace and Greedy Estimation Framework}
This link relies on reformulating the \gls{omp} and \gls{ols} selection steps using a \gls{music}-like signal-noise subspaces decomposition of $\SqrtMat{\SCovMat}_{\TgtInd}$. 
The resulting expressions are stated in Lemma~\ref{lemma:greedy_proj_sqrt}, in which the residual signal and noise subspace matrices at iteration $\TgtInd$ are defined as
\begin{equation}
\SigSpaceMat_{\TgtInd} = \CompProjMat(\GreedySet{{\AoaVec}}{\TgtInd}) \ \SigSpaceMat, \quad \NoiseSpaceMat_{\TgtInd} = \CompProjMat(\GreedySet{{\AoaVec}}{\TgtInd}) \ \NoiseSpaceMat. \label{eq:residual_subspaces}
\end{equation}
\begin{lemma} \label{lemma:greedy_proj_sqrt}
The \gls{omp} and \gls{ols} selection steps in \eqref{eq:simo_OMP_Rsqrt} and \eqref{eq:simo_OLS_Rsqrt} can be equivalently expressed as a function of the residual signal and noise subspaces as
\begin{align}
    \EstVal{\Aoa}_{\TgtInd+1}^{\textup{\Acrshort{omp}}} 
    = \argmax_{\CandVal{\Aoa}_{\TgtInd+1}} 
    \big \lVert \Hermitian{\big(\SigSpaceMat_{\TgtInd} \SqrtMat{{\SigEigMat}} + \NoiseSpaceMat_{\TgtInd}  \SqrtMat{{\NoiseEigMat}} \big)} \ \SteerVec(\CandVal{\Aoa}_{\TgtInd+1}) \big \rVert_{2}^2, \label{eq:omp_partition_evd} \\
    \EstVal{\Aoa}_{\TgtInd+1}^{\textup{\Acrshort{ols}}} 
    = \argmax_{\CandVal{\Aoa}_{\TgtInd+1}} \ 
    \frac{\big \lVert \Hermitian{\big(\SigSpaceMat_{\TgtInd} \SqrtMat{{\SigEigMat}} + \NoiseSpaceMat_{\TgtInd}  \SqrtMat{{\NoiseEigMat}} \big)} \ \SteerVec(\CandVal{\Aoa}_{\TgtInd+1}) \big \rVert_{2}^2}{\big\lVert \SteerVec_\TgtInd(\CandVal{\Aoa}_{\TgtInd+1}) \big\rVert_2^2}. \label{eq:ols_partition_evd}
\end{align}
\end{lemma}
\begin{proof}
From \eqref{eq:SVD_SCovMat} and because the signal and noise subspaces are orthogonal, $\SqrtMat{{\SCovMat}} = {\SigSpaceMat} \SqrtMat{{\SigEigMat}} + {\NoiseSpaceMat}  \SqrtMat{{\NoiseEigMat}}$.  
\end{proof}

Building on Lemma~\ref{lemma:greedy_proj_sqrt}, we approximate the common expressions of the \gls{omp} objective function in \eqref{eq:omp_partition_evd} and of the numerator of the \gls{ols} objective function in \eqref{eq:ols_partition_evd} through a low-rank argument as follows
\begin{align*}
& \big \lVert \Hermitian{\big(\SigSpaceMat_{\TgtInd} \SqrtMat{{\SigEigMat}} + \NoiseSpaceMat_{\TgtInd}  \SqrtMat{{\NoiseEigMat}} \big)} \ \SteerVec(\CandVal{\Aoa}_{\TgtInd+1}) \big \rVert_{2}^2 \\
& \stackrel{(a)}{=} \big \lVert \Hermitian{(\SigSpaceMat_\TgtInd \SqrtMat{\SigEigMat})}  \ \SteerVec (\CandVal{\Aoa}_{\TgtInd+1})\big \rVert_2^2 + \big \lVert \Hermitian{(\NoiseSpaceMat_\TgtInd \SqrtMat{\NoiseEigMat})}  \ \SteerVec (\CandVal{\Aoa}_{\TgtInd+1})\big \rVert_2^2,\\
& \stackrel{(b)}{\approx} \big \lVert \Hermitian{(\SigSpaceMat_\TgtInd \SqrtMat{\SigEigMat})}  \ \SteerVec (\CandVal{\Aoa}_{\TgtInd+1})\big \rVert_2^2 + 0. \numberthis \label{eq:omp_wmusic}
\end{align*}    
Step (a) follows from the orthogonality of the signal and noise subspaces, which cancels the cross terms when the norm in \eqref{eq:omp_partition_evd} is expanded, and step (b) neglects the noise-subspace term. 
This is justified because noise-subspace eigenvalues are typically much smaller than signal-subspace eigenvalues and because target steering vectors have near-zero projection onto the noise subspace.
Consequently, evaluating the resulting expression in \eqref{eq:omp_wmusic} for all candidate \gls{doa} at each step $\TgtInd$ corresponds to a \gls{wmusic} pseudospectrum evaluation, as in \eqref{eq:pseudo_spectrum_signal_weighted}, but based on the residual signal subspace.

\subsection{Proposed \acrshort{gimusic} Selection Step}
From this derivation, we can define two weighted variants of the proposed \gls{gimusic} algorithms, which we refer to as \gls{ompiwmusic} and \gls{olsiwmusic}, by applying the low-rank approximation in \eqref{eq:omp_wmusic} to the objective functions of \gls{omp} and \gls{ols} in \eqref{eq:omp_partition_evd} and \eqref{eq:ols_partition_evd}. 

These algorithms can be interpreted as greedy methods that, at each iteration, use a \gls{wmusic} pseudospectrum evaluation based on the residual signal subspace to select the next target \gls{doa}. 
Their resulting selection rules are stated in Proposition~\ref{prop:iwmusic_omp_ols}, while their practical implementation steps are detailed in Section~\ref{sec:imusic_practical_steps}.

\begin{proposition} \label{prop:iwmusic_omp_ols}
The \gls{ompiwmusic} and \gls{olsiwmusic} selection steps are respectively given by
\begin{align}
    \EstVal{\Aoa}_{\TgtInd+1}^{\hspace{0.20em} \textup{\acrshort{ompiwmusic}}} 
    & = \argmax_{\CandVal{\Aoa}_{\TgtInd+1}} \ \big \lVert \Hermitian{(\SigSpaceMat_{\TgtInd} \SqrtMat{{\SigEigMat}})} \ \SteerVec(\CandVal{\Aoa}_{\TgtInd+1}) \big \rVert_2^2, \label{eq:simo_OMP_iwmusic_Signal} \\
    \EstVal{\Aoa}_{\TgtInd+1}^{\hspace{0.20em} \textup{\acrshort{olsiwmusic}}} 
    & = \argmax_{\CandVal{\Aoa}_{\TgtInd+1}} \ 
    \frac{\big \lVert \Hermitian{(\SigSpaceMat_{\TgtInd} \SqrtMat{{\SigEigMat}})} \ \SteerVec(\CandVal{\Aoa}_{\TgtInd+1}) \big \rVert_2^2}{\big\lVert \SteerVec_\TgtInd(\CandVal{\Aoa}_{\TgtInd+1}) \big\rVert_2^2}. \label{eq:simo_OLS_iwmusic_Signal}
\end{align}
These selection steps can also be expressed in the residual noise-subspace form as in \eqref{eq:pseudo_spectrum_noise_weighted}.
However, as discussed for \gls{wmusic}, the signal-subspace and noise-subspace forms are not equivalent.
\end{proposition}
In contrast, empirical evidence suggests that unweighted \gls{music} often outperforms its weighted counterpart \cite{61541}.
Motivated by this observation, we further simplify \eqref{eq:omp_wmusic} as follows
\begin{align*}
& \big \lVert \Hermitian{(\SigSpaceMat_\TgtInd \SqrtMat{\SigEigMat})}  \ \SteerVec (\CandVal{\Aoa}_{\TgtInd+1})\big \rVert_2^2 
\stackrel{(c)}{\approx} \big \lVert \Hermitian{\SigSpaceMat}_\TgtInd \  \SteerVec (\CandVal{\Aoa}_{\TgtInd+1})\big \rVert_2^2, \numberthis \label{eq:omp_music_Sig} \\
& \hspace{60pt} \stackrel{(d)}{=} \big\lVert \SteerVec_\TgtInd(\CandVal{\Aoa}_{\TgtInd+1}) \big\rVert_2^2 - \big \lVert \Hermitian{\NoiseSpaceMat}_\TgtInd \  \SteerVec (\CandVal{\Aoa}_{\TgtInd+1})\big \rVert_2^2. \numberthis \label{eq:omp_music_Noise}
\end{align*}
Step (c) removes the signal-eigenvalue weighting, and step (d) rewrites the resulting expression in terms of the residual noise subspace by applying 
$\SigSpaceMat_\TgtInd \Hermitian{\SigSpaceMat}_\TgtInd = \CompProjMat(\GreedySet{\EstVal{\AoaVec}}{k}) - \NoiseSpaceMat_\TgtInd \Hermitian{\NoiseSpaceMat}_\TgtInd,$ and then expanding the norm.
Observe now that evaluating the resulting expression in \eqref{eq:omp_music_Sig} for all candidate \gls{doa}, at each step $\TgtInd$, corresponds to a \gls{music} pseudospectrum evaluation, as in \eqref{eq:pseudo_spectrum_signal_weighted}, but based on the residual signal subspace.
Similarly, the second term in \eqref{eq:omp_music_Noise} corresponds to a \gls{music} pseudospectrum evaluation based on the residual noise subspace, as in \eqref{eq:pseudo_spectrum_noise_weighted}.

Consequently, we can define two additional variants of the proposed \gls{gimusic} algorithms, which we refer to as \gls{ompimusic} and \gls{olsimusic}, by applying the unweighted low-rank approximation in \eqref{eq:omp_music_Sig} or \eqref{eq:omp_music_Noise}. 

These algorithms can be interpreted as greedy methods that, at each iteration, use a \gls{music} pseudospectrum evaluation based on either the residual signal or noise subspace to select the next target \gls{doa}.
Their resulting selection rules are stated in Proposition~\ref{prop:imusic_omp_ols}, while their practical implementation steps are detailed next.
\begin{proposition} \label{prop:imusic_omp_ols}
The \gls{ompimusic} and \gls{olsimusic} selection steps based on the residual signal and noise subspace are respectively defined as
\begin{align} 
    \EstVal{\Aoa}_{\TgtInd+1}^{\textup{\acrshort{ompimusic}}} 
    & = \argmax_{\CandVal{\Aoa}_{\TgtInd+1}} \ \big \lVert \Hermitian{\SigSpaceMat}_{\TgtInd} \ \SteerVec(\CandVal{\Aoa}_{\TgtInd+1}) \big \rVert_2^2, \label{eq:simo_OMP_imusic_Signal} \\ 
     = & \argmax_{\CandVal{\Aoa}_{\TgtInd+1}} \big\lVert \SteerVec_\TgtInd(\CandVal{\Aoa}_{\TgtInd+1}) \big\rVert_2^2 - \big \lVert \Hermitian{\NoiseSpaceMat}_{\TgtInd} \SteerVec(\CandVal{\Aoa}_{\TgtInd+1}) \big \rVert_2^2, \label{eq:simo_OMP_imusic_Noise} \\
    \EstVal{\Aoa}_{\TgtInd+1}^{ \textup{\acrshort{olsimusic}}} & = \argmax_{\CandVal{\Aoa}_{\TgtInd+1}} \ 
    \frac{\big \lVert \Hermitian{\SigSpaceMat}_{\TgtInd} \ \SteerVec(\CandVal{\Aoa}_{\TgtInd+1}) \big \rVert_2^2}{\big\lVert \SteerVec_\TgtInd(\CandVal{\Aoa}_{\TgtInd+1}) \big\rVert_2^2}, \label{eq:simo_OLS_imusic_Signal} \\
    & = \argmax_{\CandVal{\Aoa}_{\TgtInd+1}} \ 1- 
    \frac{\big \lVert \Hermitian{\NoiseSpaceMat}_{\TgtInd} \ \SteerVec(\CandVal{\Aoa}_{\TgtInd+1}) \big \rVert_2^2}{\big\lVert \SteerVec_\TgtInd(\CandVal{\Aoa}_{\TgtInd+1}) \big\rVert_2^2}. \label{eq:simo_OLS_imusic_Noise}
\end{align}
The signal and noise subspace forms are equivalent in the sense that they yield the same selection result.
However, the noise-subspace form in \eqref{eq:simo_OMP_imusic_Noise} is less attractive in practice because it introduces an additional operation for the first term compared with \eqref{eq:simo_OMP_imusic_Signal}.
\end{proposition}
In conclusion, the proposed \gls{ompimusic} and \gls{olsimusic} selection steps in \eqref{eq:simo_OMP_imusic_Signal} and \eqref{eq:simo_OMP_imusic_Noise} are a direct application of the \gls{music} and \gls{wmusic} pseudospectrum evaluation based on the residual signal subspace, respectively in \eqref{eq:pseudo_spectrum_signal} and \eqref{eq:pseudo_spectrum_signal_weighted}. 
The \gls{ols} variants in \eqref{eq:simo_OLS_imusic_Signal} and \eqref{eq:simo_OLS_imusic_Noise} are a modified versions of the same criteria, as they include an additional normalization term in the denominator, which improves the estimation performance at the price of an additional computation.  

\subsection{Practical \gls{gimusic} Implementation} \label{sec:imusic_practical_steps}

Algorithm~\ref{alg:greedy_imusic} summarizes the successive steps of the proposed \gls{gimusic} methods for multi-target \gls{doa} estimation.
To clarify this workflow, the following subsections detail each processing step.

\begin{algorithm}[t]
    \caption{The proposed \gls{gimusic} algorithms}
    \label{alg:greedy_imusic}
    \KwIn{$\TgtAll$ and $\SRxSigMat$ or $\SCovMat$}
    \vspace{0.05cm}
    \KwOut{$\widehat{\AoaVec}_{\TgtAll} = [\widehat{\Aoa}_1 \ \dots  \ \widehat{\Aoa}_{\TgtAll}]$}

    \Begin{
    \textbf{1. \Acrshort{evd}:} Compute eigenvalues/eigenvectors of $\SCovMat$\;
    \textbf{2. Signal/Noise partition:} $\SCovMat = \SigSpaceMat \SigEigMat \Hermitian{\SigSpaceMat} + \NoiseSpaceMat \NoiseEigMat \Hermitian{\NoiseSpaceMat}$\;
    \textbf{3. Initialization:} $\GreedySet{\EstVal{\AoaVec}}{0} \gets [~]$, $\ProjMat^{\Comp} \big(\GreedySet{\EstVal{\AoaVec}}{0}\big) \gets \Identity_\SRxAll$ and $\SigSpaceMat_0 \gets \SigSpaceMat$ or $\NoiseSpaceMat_0 \gets \NoiseSpaceMat$\;
        \For{$k\gets 0$ \KwTo $\TgtAll-1$}{
        	\textbf{4. Selection Step:} 
            \newline \vspace{2pt} a) \gls{ompiwmusic}: $\EstVal{\Aoa}_{\TgtInd+1} \gets $ \eqref{eq:simo_OMP_iwmusic_Signal}; 
            \newline \vspace{2pt} b) \gls{olsiwmusic}: $~\EstVal{\Aoa}_{\TgtInd+1} \gets $ \eqref{eq:simo_OLS_iwmusic_Signal}; 
            \newline \vspace{2pt} c) \gls{ompimusic}: $\EstVal{\Aoa}_{\TgtInd+1} \gets $ \eqref{eq:simo_OMP_imusic_Signal};
             \newline \vspace{2pt} d) \gls{olsimusic}: $~\EstVal{\Aoa}_{\TgtInd+1} \gets $ \eqref{eq:simo_OLS_imusic_Signal} or \eqref{eq:simo_OLS_imusic_Noise}\;

        	\textbf{5. Update Step:} $\GreedySet{\EstVal{\AoaVec}}{\TgtInd+1} \gets [\GreedySet{\EstVal{\AoaVec}}{\TgtInd},\EstVal{\Aoa}_{\TgtInd+1}]$, $\ProjMat^{\Comp}\big(\GreedySet{\EstVal{\AoaVec}}{\TgtInd+1}\big) \gets $ \eqref{eq:simo_ProjCompMat} and $\SigSpaceMat_{\TgtInd} \gets \ProjMat^{\Comp}(\GreedySet{{\AoaVec}}{\TgtInd}) \ \SigSpaceMat$ or $\NoiseSpaceMat_{\TgtInd} \gets \ProjMat^{\Comp}(\GreedySet{{\AoaVec}}{\TgtInd}) \ \NoiseSpaceMat$\; 
    	}
    }
\end{algorithm}

\subsubsection{Eigenvalue Decomposition}
The first step computes the \gls{evd} of the sample covariance matrix $\SCovMat$.
Unlike prior \gls{imusic} approaches that require an \gls{evd} at each iteration, the proposed algorithms perform only one \gls{evd} at initialization to obtain the signal and noise subspaces. Then, the residual subspaces are updated at each iteration through simple matrix multiplications with the projection matrix $\ProjMat^{\Comp}(\GreedySet{{\AoaVec}}{\TgtInd})$, as in \eqref{eq:residual_subspaces}.

\subsubsection{Signal- and Noise-Subspace Partition}
Using the computed eigenpairs, we next partition $\SCovMat$ into signal and noise components, as in Section~\ref{sec:music_sig_noise_partition}.

\subsubsection{Initialization}
Initialization follows the standard greedy methods procedure detailed in Section~\ref{sec:OMP_OLS}.
Depending on the selected criterion, we additionally initialize one residual subspace matrix, i.e., $\SigSpaceMat_0 = \SigSpaceMat$ or $\NoiseSpaceMat_0 = \NoiseSpaceMat$.
The algorithm then alternates between selection and update steps until all targets are estimated.

\subsubsection{Selection Step}
At each iteration, the algorithm selects the next target \gls{doa} by solving an optimization problem over the search grid $\Grid_\Aoa$.
Optionally, gradient ascent refinement can enable off-grid target detection.
Depending on the variant of interest, the selection step uses the criteria defined in Proposition~\ref{prop:imusic_omp_ols} or Proposition~\ref{prop:iwmusic_omp_ols}.

\subsubsection{Update Step}
After selection, the update step is described in Section~\ref{sec:OMP_OLS}.
In parallel, we update the residual signal or noise subspace matrices at each iteration from \eqref{eq:residual_subspaces}.

\section{Fast Fourier Transform Acceleration} \label{sec:fft_acceleration}
This section presents \gls{fft}-based accelerations for the algorithms introduced in Sections~\ref{sec:state_of_the_art} and \ref{sec:imusic} in the \gls{ula} setting.

For each algorithm, the selection step evaluates an objective function over a grid $\Grid_\Aoa$ of candidate \gls{doa} values.
We denote by $\PseudoVec \in \Rset^{1 \times \GridAll}$ the vector containing these evaluations for all $\GridAll$ grid points.
In general, $\PseudoVec$ is obtained by replacing the steering vector $\SteerVec(\Aoa_{\GridInd})$ in the corresponding objective function with the steering matrix $\SteerMat(\Grid_\Aoa)$ and then applying the column-wise $\ell_2$ norm $\lVert \cdot \rVert_{2,\Col}$.
When the array is a \gls{ula}, $\SteerMat(\Grid_\Aoa)$ admits a \gls{dft} representation, which can be computed efficiently using the \gls{fft}.
Accordingly, we define the \gls{dft} matrix $\FFTMat \in \Cset^{\SRxAll \times \GridAll}$ with entries $\FFTVal_{\SRxInd,\GridInd}= e^{-j 2 \pi \SRxInd \GridInd/\GridAll}$ for antenna indices $\SRxInd=0,\dots,\SRxAll-1$ and normalized frequency indices $\GridInd=0,\dots,\GridAll-1$.
These normalized frequencies are then mapped to their corresponding \gls{doa} values in $\Grid_\Aoa$.

With this notation in place, the \gls{music} (respectively, \gls{wmusic}) pseudospectrum evaluations based on the noise subspace in \eqref{eq:pseudo_spectrum_noise} (respectively, \eqref{eq:pseudo_spectrum_noise_weighted}) and on the signal subspace in \eqref{eq:pseudo_spectrum_signal} (respectively, \eqref{eq:pseudo_spectrum_signal_weighted}) can be written equivalently as
\begin{align}
    \PseudoVec_{\NoiseSpaceMat}^\text{\gls{music}} & = \left\lVert \Hermitian{\NoiseSpaceMat} \ \FFTMat \right\rVert_{2,\Col}^{-2}, \ \PseudoVec_{\NoiseEigMat\NoiseSpaceMat}^\text{\gls{wmusic}} = \left\lVert \Hermitian{(\SqrtMat{\NoiseEigMat}\NoiseSpaceMat)} \ \FFTMat \right\rVert_{2,\Col}^{-2}, \\
    \PseudoVec_{\SigSpaceMat}^\text{\gls{music}} & = \left\lVert \Hermitian{\SigSpaceMat} \ \FFTMat \right\rVert_{2,\Col}^2, \ \PseudoVec_{\SigEigMat\SigSpaceMat}^\text{\gls{wmusic}} = \left\lVert \Hermitian{(\SqrtMat{\SigEigMat}\SigSpaceMat)} \ \FFTMat \right\rVert_{2,\Col}^2.
\end{align}
Similarly, at each iteration $\TgtInd$, the \gls{omp} objective values over the grid can be computed from the residual observation matrix $\SRxSigMat_{\TgtInd}$ in \eqref{eq:simo_OMP_Y} and from the residual sample covariance matrix $\GreedySet{\SCovMat}{\TgtInd}$ in \eqref{eq:simo_OMP_Rsqrt} as
\begin{equation}
    \PseudoVec_{\SRxSigMat_{\TgtInd}}^\text{\gls{omp}} = \lVert \Hermitian{\SRxSigMat}_{\TgtInd} \ \FFTMat \big\rVert_{2,\Col}^2, \ \PseudoVec_{\SCovMat_{\TgtInd}}^\text{\gls{omp}} = \lVert \Hermitian{(\SqrtMat{\SCovMat}_{\TgtInd})} \ \FFTMat \big\rVert_{2,\Col}^2.
\end{equation}
For \gls{music} and \gls{omp}, such \gls{fft}-based accelerations are already known in the literature.
For \gls{ols}, in contrast, a \gls{dft} formulation of the objective function is available only by exploiting the projection-matrix decomposition in \cite{7543,10979421}.
Building on that decomposition, Proposition~\ref{proposition:fft_ols} provides an \gls{fft}-based acceleration of \gls{ols}.
\begin{proposition} \label{proposition:fft_ols}
In the case of \glspl{ula}, the evaluation of the \gls{ols} objective functions in \eqref{eq:simo_OLS_Y} and \eqref{eq:simo_OLS_Rsqrt} can be accelerated through the following element-wise \gls{fft} quotients, respectively:
\begin{equation}
    \PseudoVec_{\SRxSigMat_{\TgtInd}}^\textup{\gls{ols}} = \frac{\big\lVert \Hermitian{\SRxSigMat}_{\TgtInd} \  \FFTMat \big\rVert_{2,\Col}^2}{\big\lVert  \ProjMat^{\Comp} \big(\GreedySet{\EstVal{\AoaVec}}{\TgtInd}\big) \ \FFTMat \big\rVert_{2,\Col}^2}, \
    \PseudoVec_{\SCovMat_{\TgtInd}}^\textup{\gls{ols}} = \frac{\big\lVert \Hermitian{(\SqrtMat{\SCovMat}_{\TgtInd})} \ \FFTMat \big\rVert_{2,\Col}^2}{\big\lVert  \ProjMat^{\Comp} \big(\GreedySet{\EstVal{\AoaVec}}{\TgtInd}\big) \ \FFTMat \big\rVert_{2,\Col}^2}.
\end{equation}
\end{proposition}
Having established the \gls{ols} case, we now turn to the proposed \gls{gimusic} methods.
Proposition~\ref{proposition:fft_imusic} presents the \gls{fft}-based accelerations of \gls{ompimusic} and \gls{olsimusic}.
\begin{proposition} \label{proposition:fft_imusic}
In the case of \glspl{ula}, the computation of the \gls{ompimusic} objective function in \eqref{eq:simo_OMP_imusic_Signal} can be accelerated using an \gls{fft} as
\begin{equation}
    \PseudoVec_{\SigSpaceMat_{\TgtInd}}^\textup{\gls{ompimusic}} = \big\lVert \Hermitian{\SigSpaceMat}_{\TgtInd} \  \FFTMat \big\rVert_{2,\Col}^2.
\end{equation}
The evaluation of the \gls{olsimusic} objective functions in \eqref{eq:simo_OLS_imusic_Signal} and \eqref{eq:simo_OLS_imusic_Noise} can be accelerated through the following element-wise \gls{fft} quotients, respectively:
\begin{align}
    \PseudoVec_{\SigSpaceMat_{\TgtInd}}^\textup{\gls{olsimusic}} & = \frac{\big\lVert \Hermitian{\SigSpaceMat}_{\TgtInd} \  \FFTMat \big\rVert_{2,\Col}^2}{\big\lVert  \ProjMat^{\Comp} \big(\GreedySet{\EstVal{\AoaVec}}{\TgtInd}\big) \FFTMat \big\rVert_{2,\Col}^2}, \\  
    \PseudoVec_{\NoiseSpaceMat_{\TgtInd}}^\textup{\gls{olsimusic}} & = 1 - \frac{\big\lVert \Hermitian{\NoiseSpaceMat}_{\TgtInd} \ \FFTMat \big\rVert_{2,\Col}^2}{\big\lVert  \ProjMat^{\Comp} \big(\GreedySet{\EstVal{\AoaVec}}{\TgtInd}\big) \ \FFTMat \ big\rVert_{2,\Col}^2}.
\end{align}
\end{proposition}
By the same argument, weighted variants are accelerated by replacing the signal and noise subspace matrices with their weighted counterparts, as in \gls{wmusic}.
\section{Computational Complexity Analysis} \label{sec:Complexity Analysis}

\begin{table*}[t]
\centering
\caption{Theoretical complexity comparison of the main processing steps of the considered multi-target \Acrshort{doa} estimation methods without \gls{fft} acceleration. In case of a \gls{ula}, the \gls{fft} acceleration reduces the complexity of the objective function evaluation by replacing $\GridAll \SRxAll$ by $\GridAll \log_2(\SRxAll)$.
For conciseness, we define $\TgtAll_\SRxAll = \min(\TgtAll,\SRxAll-\TgtAll)$.}
\label{tab:complexity_comparison}
\renewcommand{\arraystretch}{1.2}
\setlength{\tabcolsep}{6pt}
\begin{tabular}{l|l|l|l|l|l|l|l}
\toprule
 &
\textbf{\gls{music}} &
\textbf{\gls{omp}-$\SRxSigMat$} &
\textbf{\gls{ols}-$\SRxSigMat$} &
\textbf{\gls{omp}-$\SqrtMat{\SCovMat}$} &
\textbf{\gls{ols}-$\SqrtMat{\SCovMat}$} &
\textbf{\gls{ompimusic}} &
\textbf{\gls{olsimusic}} \\
\midrule
$\SCovMat=\frac{1}{\ObsAll \SubAll} \SRxSigMat \Hermitian{\SRxSigMat}$ & $\bigO(\ObsAll \SubAll \SRxAll^2)$ & - & - & $\bigO(\ObsAll \SubAll \SRxAll^2)$ & $\bigO(\ObsAll \SubAll \SRxAll^2)$ & $\bigO(\ObsAll \SubAll \SRxAll^2)$ & $\bigO(\ObsAll \SubAll \SRxAll^2)$ \\
\midrule
\gls{evd} of $\SCovMat$ & $\bigO(\SRxAll^3)$ & - & - & - & - & $\bigO(\SRxAll^3)$ & $\bigO(\SRxAll^3)$ \\
\midrule
\makecell[l]{Numerator of \\ objective function} & $\bigO(\GridAll \SRxAll \TgtAll_\SRxAll)$ & $\TgtAll \ \bigO(\GridAll \SRxAll \ObsAll \SubAll)$ & $\TgtAll \ \bigO(\GridAll \SRxAll \ObsAll \SubAll)$ & $\TgtAll \ \bigO(\GridAll \SRxAll^2)$ & $\TgtAll \ \bigO(\GridAll \SRxAll^2)$ & $\TgtAll \ \bigO(\GridAll \SRxAll \TgtAll)$ & $\TgtAll \ \bigO(\GridAll \SRxAll \TgtAll_\SRxAll)$ \\
\midrule
\makecell[l]{Denominator of \\ objective function} & - & - & $\TgtAll \ \bigO(\GridAll \SRxAll^2)$ & - & $\TgtAll \ \bigO(\GridAll \SRxAll^2)$ & - & $\TgtAll \ \bigO(\GridAll \SRxAll^2)$ \\
\midrule
Peaks detection & $\bigO(\GridAll)$ & - & - & - & - & - & - \\
\bottomrule
\end{tabular}
\end{table*}

In this section, we present the asymptotic computational complexity of the algorithms presented in Sections~\ref{sec:state_of_the_art} and~\ref{sec:imusic}. 
Table~\ref{tab:complexity_comparison} summarizes the number of \glspl{flop} for the main processing steps of each algorithm. 
Unless stated otherwise, these counts are expressed as a function of key system parameters by following standard \gls{flop} models for basic matrix operations \cite{Richards2013PrinciplesOM,625604}. 
The number of flops in Table~\ref{tab:complexity_comparison} is provided without the \gls{fft} acceleration; its benefit is detailed next.

Assuming each method takes the observation matrix $\SRxSigMat$ as input, computing the sample covariance matrix $\SCovMat$ requires $\bigO(\ObsAll\SubAll\SRxAll^2)$ operations.
For \gls{music}-based methods, computing the \gls{evd} of $\SCovMat$ requires $\bigO(\SRxAll^3)$ \glspl{flop} with standard routines (e.g., the QZ algorithm) \cite{golub2013matrix}.
In addition, to exploit \gls{fft} acceleration in \gls{omp}/\gls{ols} methods based on $\SCovMat$, one initial square-root matrix computation is required (for instance via \gls{evd}). 
We next detail the complexity of the selection step for each algorithm.
Among the considered methods, \gls{music} is the only non-iterative one; hence, it requires a single pseudo-spectrum evaluation over the full grid.
As shown in Section~\ref{sec:OMP_OLS}, the numerators of the \gls{omp} and \gls{ols} selection steps are equivalent.
However, \gls{ols} includes an additional denominator term.
For \gls{gimusic}-based algorithms, \gls{olsimusic} can evaluate its numerator in either the signal or the noise subspace (Proposition~\ref{prop:imusic_omp_ols}), which leads to different complexity expressions.
When a \gls{ula} is used, all methods can exploit \gls{fft} acceleration for numerator and denominator evaluations.
This reduces the $\GridAll\SRxAll$ factor to $\GridAll\log(\SRxAll)$ \cite{donciuFFT2025}.
Finally, because \gls{music} detects all targets in one stage, it requires pseudo-spectrum peak detection.
This adds $\bigO(\GridAll)$ operations, whereas greedy methods only require maximum-value selection at each iteration.
Note that the weighting of the \gls{wmusic} variants is not displayed in Table~\ref{tab:complexity_comparison}, but they do not change the asymptotic complexity of the selection step compared with their unweighted counterparts.

Overall, \gls{music} is expected to have the lowest computational complexity as it does not require iterative processing. 
However, it requires a peak-detection step that greedy methods do not require.
Traditional \gls{omp} algorithms exhibit lower complexity than \gls{ols} due to the absence of a denominator term, although this gap is substantially reduced compared with naive \gls{ols} implementations thanks to the projection-matrix decomposition in \eqref{eq:proj_mat_decomposition}.
The implementation based on $\SqrtMat{\SCovMat}$ is usually preferred over the one based on $\SRxSigMat$ because $\ObsAll\SubAll \gg \SRxAll$ in the considered scenario.
Moreover, the presented \gls{gimusic} algorithms have lower complexity than their traditional \gls{omp} and \gls{ols} counterparts as they project steering vectors onto either the residual signal or noise subspaces rather than $\SqrtMat{\SCovMat}_{\TgtInd}$.
Thanks to the possible evaluation of the \gls{olsimusic} numerator using either the signal or noise subspace, its additional computational burden with respect to \gls{ompimusic} can be reduced, depending on the relative sizes of $\TgtAll$ and $\SRxAll$. 
Finally, \gls{fft} acceleration further reduces the computational burden of all algorithms in \gls{ula} configurations.
While the asymptotic analysis in Table~\ref{tab:complexity_comparison} identifies the dominant computational components of each method, it does not fully capture practical runtime behavior.
Therefore, Section~\ref{sec:simulations} complements this analysis with processing-time evaluations.
\section{Simulation Setup and Results} \label{sec:simulations}
In this section, we compare the proposed \gls{gimusic} algorithms with \gls{music}, \gls{omp}, and \gls{ols} for multi-target \gls{doa} estimation in a \gls{pr} system.

\subsection{Simulation Scenario}
We evaluate performance over $10{,}000$ Monte Carlo simulations in a scenario with a colocated \gls{ap} and \gls{pr}.
In each simulation, we consider a fixed setup in which multiple targets are randomly positioned within a maximum range of $\Range_{max}=60$ m.
We assume half-wavelength antenna spacing and optimal orientation toward the coverage area. 
Following the Wi-Fi 7 standard, we set the carrier frequency to $f_c=5$ GHz, and the subcarrier spacing to $\SubSpacing=78.125$ kHz \cite{11090080}.
Unless stated otherwise, the system parameters are $\TgtAll=8$ targets, $\SRxAll=16$ antennas, $Q=512$ subcarriers (i.e., $40$ MHz), $\ObsAll=10$ successive symbols, and an \gls{snr} of $40$ dB.
We compute the mean \gls{snr} by averaging \eqref{eq:snr} across all simulations.
With this definition, the \gls{snr} is dominated by the strongest target, which is also the easiest to detect.
As a result, the radar operating \gls{snr} regime may appear higher than that of a conventional communication system.
We define a search grid $\Grid_\Aoa$ with $\GridAll=2048$ uniformly spaced points between $-1$ and $1$, representing the normalized angular domain.

\subsection{Estimation of the number of targets}
In the previous sections, the number of targets $\TgtAll$ is assumed to be known. 
In practice, however, this quantity is unknown and must be estimated from the observed data.
Target-number estimation can be formulated as a model-order selection problem, which is typically addressed using \gls{itc}.
These criteria minimize an objective function that balances a likelihood term against a complexity penalty.
In this work, we consider the \gls{aic} criterion of Wax and Kailath \cite{1164557}.
Intuitively, this criterion evaluates the number of targets by assessing the rank of the signal subspace, which is determined by the number of dominant eigenvalues of the sample covariance matrix $\SCovMat$.
This rank-based criterion is widely used in the \gls{doa} estimation literature and is independent of the specific estimation algorithm from Section~\ref{sec:state_of_the_art} and \ref{sec:imusic}, making it suitable for a fair comparison. 
In this disjoint approach, the number of targets $\EstVal{\TgtAll}_\text{Rank}$ is first estimated before applying the \gls{doa} estimation algorithm.
In our previous work \cite{willameStopping}, we proposed a hybrid criterion that combines the eigenvalue-based \gls{aic} of Wax and Kailath \cite{1164557} with a selection-based criterion that depends on detected \gls{ols} peaks.
Specifically, the first $\EstVal{\TgtAll}_\text{Rank}$ estimated angles are systematically added to the set of detected peaks $\EstVal{\AoaVec}$. Then, the hybrid criterion stops adding new angles when it does not improve the proposed \gls{aic}-\gls{ols} criterion, yielding $\EstVal{\TgtAll}_\text{Hybrid} \geq \EstVal{\TgtAll}_\text{Rank}$ detected targets. 
This hybrid approach significantly improves \gls{ols} and \gls{olsimusic}, but it does not provide a significant gain for \gls{music}, \gls{omp}, or \gls{ompimusic}.
Accordingly, for the two \gls{ols} variants, we also report results obtained with the hybrid criterion.

\subsection{Diagnostic Metrics}
To interpret the observed performance trends beyond the considered scenario, we introduce two scenario-agnostic diagnostic metrics: the \textit{steering-vector correlation metric} ($\ResRatio$) and the \textit{signal correlation metric} ($\CorrRatio$).
Together, these metrics provide a compact measure of the difficulty of the detection scenario.

The steering-vector correlation metric $\ResRatio$ quantifies the mean angular proximity between targets, which captures the physical resolution of the system.
The signal correlation metric $\CorrRatio$ quantifies the mean correlation between target signals, which can be influenced by their intrinsic correlation or by their temporal and spectral diversity.
In the considered passive \gls{ofdm} radar scenario, the target signals are correlated because they share the same unknown data symbols, but they become less correlated as $\SubAll$ and $\ObsAll$ increase as the corresponding range and Doppler resolutions improve.

To define these metrics, we introduce the steering-vector resolution matrix $\ResMat \in \Cset^{\TgtAll \times \TgtAll}$ and the signal correlation matrix $\CorrMat \in \Cset^{\TgtAll \times \TgtAll}$, which are defined as the normalized Gram matrices of the steering vectors and of the target signals, respectively,
\begin{equation} \label{eq:Corr_Res_Matrices}
   \ResMat = \frac{1}{\SRxAll} \Hermitian{\SteerMat}(\AoaVec) \SteerMat(\AoaVec),
   \quad \text{and} \quad
   \CorrMat = \frac{1}{\ObsAll \SubAll} \CoefMat^{\prime} \Hermitian{\CoefMat^{\prime}}.
\end{equation}
The closer these matrices are to diagonal, the lower the inter-target correlations and, consequently, the more favorable the detection scenario.
Their diagonality can be quantified using the criterion proposed in \cite{ALYANI2017290}, which evaluates the mean ratio between diagonal and off-diagonal entries of a matrix, yielding a score between $0$ (for a balanced matrix) and $1$ (for a perfectly diagonal matrix).
The proposed diagnostic metrics $\ResRatio$ and $\CorrRatio$ are therefore computed as the average diagonality ratings across the Monte Carlo simulations.

In theory, both $\ResRatio$ and $\CorrRatio$ should correlate positively with detection performance.
More specifically, greedy methods such as \gls{omp} and \gls{ols} are expected to perform better when $\ResRatio$ is high, because they rely on the correlation between the residual and the steering vectors to select new peaks.
When $\CorrRatio$ is low, greedy methods become sensitive to early selection errors due to unresolved targets \cite{tropp_greed_2004}.
By contrast, subspace methods such as \gls{music} are expected to perform better when $\CorrRatio$ is high, because their super-resolution capability relies on the separability of the targets in the domains averaged in the sample covariance matrix computation, which are here the time and frequency domains \cite{pesavento_three_2023}.

\subsection{Performance Analysis}
The next subsections analyze detection, precision, and timing performance.
For the metric evaluations, we associate detected \glspl{doa} with true \glspl{doa} using the Hungarian algorithm \cite{Kuhn2010}.
We count a hit when a detected peak lies within the main lobe of its associated true target \gls{doa}.
We declare a false alarm when a detected peak cannot be associated with any true target \gls{doa}.
For visual clarity, we omit the curves of \gls{ompiwmusic} and \gls{olsiwmusic} in the performance analysis. 
Their performance is briefly discussed at the end of this section.

\subsubsection{\textbf{Detection Performance}}
\begin{figure*}[!t]
\captionsetup[subfigure]{labelformat=empty, labelsep=none, justification=raggedright, singlelinecheck=false}

\hspace{0.05\textwidth}
\begin{center}
\subfloat[]{
\includegraphics[trim=3.4cm 9.5cm 4.3cm 9.91cm,width = 0.3\textwidth]{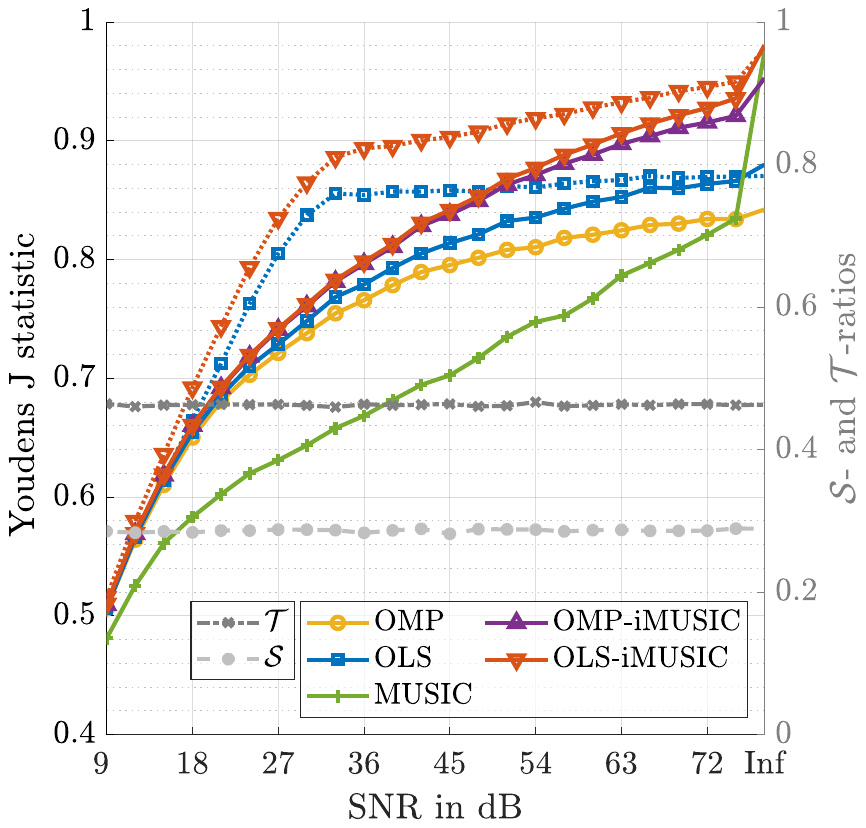}\label{fig:Jstat_vs_SNR}
}
\hfill
\subfloat[]{
\includegraphics[trim=3.4cm 9.5cm 4.3cm 9.91cm,width = 0.3\textwidth]{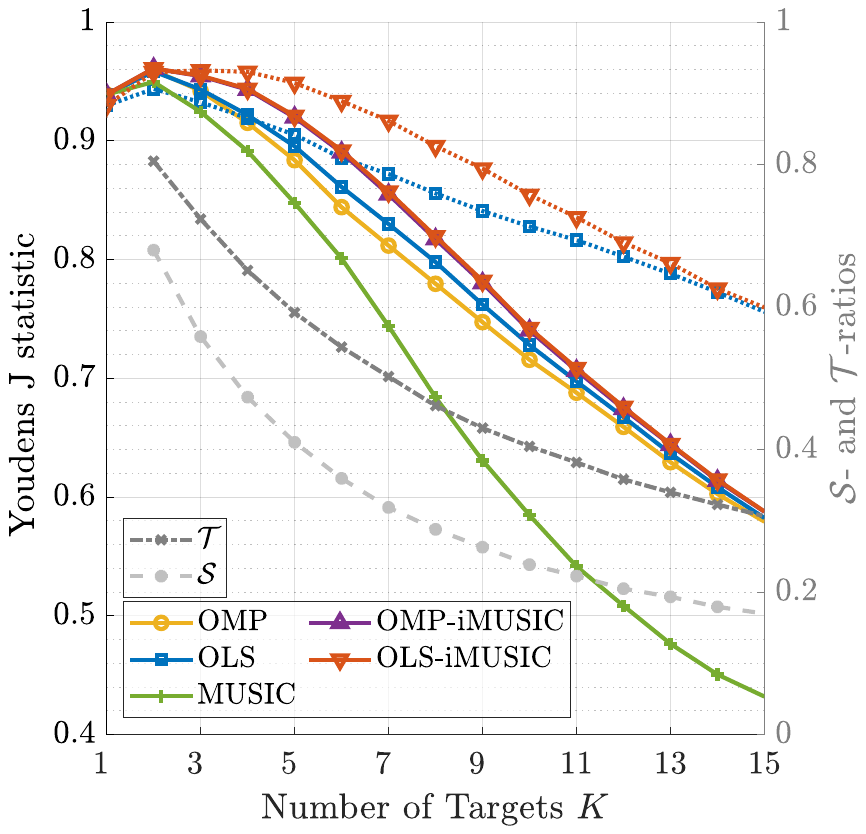}\label{fig:Jstat_vs_K} 
}
\hfill
\subfloat[]{
\includegraphics[trim=3.4cm 9.5cm 4.3cm 9.91cm,width = 0.3\textwidth]{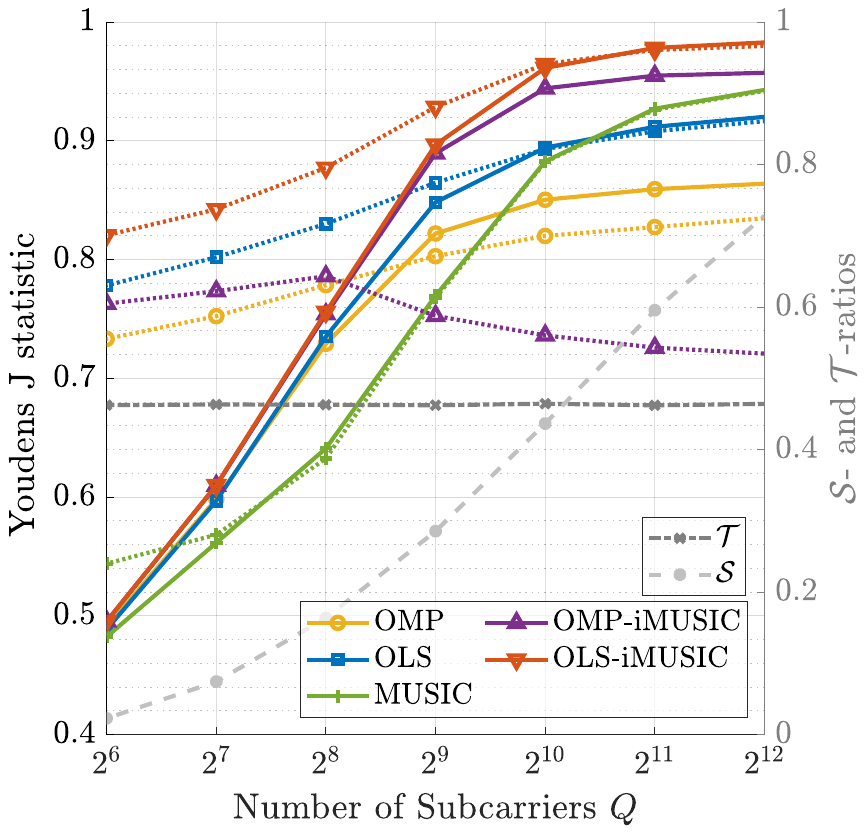}\label{fig:Jstat_vs_Q}
} 
\end{center}
\vspace*{-0.3\textwidth}
\hspace{-.01\textwidth} (a) \hspace{.305\textwidth} (b) \hspace{.305\textwidth} (c) \hfill
\vspace*{0.27\textwidth}
\caption{Youden J statistic versus (a) \gls{snr}, (b) number of targets $\TgtAll$, and (c) number of subcarriers $\SubAll$. Solid lines use the eigenvalue-based \gls{aic} stopping criterion, whereas dotted lines use the hybrid criterion. In (c), hybrid criterion curves are shown for all methods. The $\CorrRatio$ and $\ResRatio$ diagnostics are also plotted (dashed light gray and dash-dotted dark gray, respectively).}
\label{fig:Jstat_comparison}
\end{figure*}
We evaluate each algorithm's ability to detect targets using the Youden J statistic, defined as the difference between hit rate and false alarm rate \cite{Youden1950Index,MartinezCamblor2019}.
A score of $1$ indicates perfect detection, whereas a score of $0$ indicates equal hit and false alarm rates.
\figurename~\ref{fig:Jstat_comparison} reports the Youden J statistic versus \gls{snr}, the number of targets $\TgtAll$, and the number of subcarriers $\SubAll$.
Because $\SubSpacing$ is fixed, increasing $\SubAll$ also increases the bandwidth $\Bandwith$.

We first examine the \gls{snr} sweep in \figurename~\ref{fig:Jstat_vs_SNR}.
At low \gls{snr}, greedy methods outperform \gls{music}.
As the \gls{snr} increases, \gls{music} approaches greedy methods because signal and noise subspaces become easier to separate.
Accordingly, the performance gap between \gls{gimusic} and their \gls{omp}/\gls{ols} counterparts is larger at high \gls{snr}.

We next vary the number of targets in \figurename~\ref{fig:Jstat_vs_K}.
Detection performance decreases as $\TgtAll$ increases, as expected to follow the simultaneous decrease of $\CorrRatio$ and $\ResRatio$.
At low $K$, all methods show similar detection performance.
As $K$ increases, however, the gain of \gls{ols}-based variants with the hybrid criterion becomes more pronounced.
For visual clarity, we omit the hybrid-criterion curves of \gls{omp} and \gls{ompimusic}.
Although they exceed their eigenvalue-based counterparts at high $K$, they remain below \gls{ols} and \gls{olsimusic}.

\figurename~\ref{fig:Jstat_vs_Q} analyzes the effect of $\SubAll$ for both stopping criteria.
At low $\SubAll$, the eigenvalue-based criterion degrades all methods \cite{willameStopping}; therefore, we discuss this low-$\CorrRatio$ regime through the hybrid target selection criterion.
As expected, \gls{music} fails in this regime. 
In contrast, the proposed \gls{gimusic} variants outperform their \gls{omp}/\gls{ols} counterparts, indicating stronger robustness for low-$\CorrRatio$  than \gls{music}.
As $\SubAll$ increases, methods that rely on the eigenvalue-based \gls{aic} criterion improve markedly.
At high $\CorrRatio$, \gls{ompimusic} approaches \gls{olsimusic}, which is consistent with the strong behavior of \gls{music} in this regime.

Overall, the \gls{gimusic} methods consistently outperform their respective \gls{omp} and \gls{ols} baselines and also outperform \gls{music}.
Consistent with \cite{willameStopping}, the hybrid \gls{aic} criterion improves both \gls{ols} methods compared to the eigenvalue-based \gls{aic} criterion.
Moreover, the gain of \gls{olsimusic} over \gls{ompimusic} is most visible with the hybrid criterion.

\subsubsection{\textbf{Precision Performance}}
\begin{figure*}[!t]
\captionsetup[subfigure]{labelformat=empty, labelsep=none, justification=raggedright, singlelinecheck=false}

\hspace{0.05\textwidth}
\begin{center}
\subfloat[]{
\includegraphics[trim=3.4cm 9.5cm 4.3cm 9.91cm,width = 0.3\textwidth]{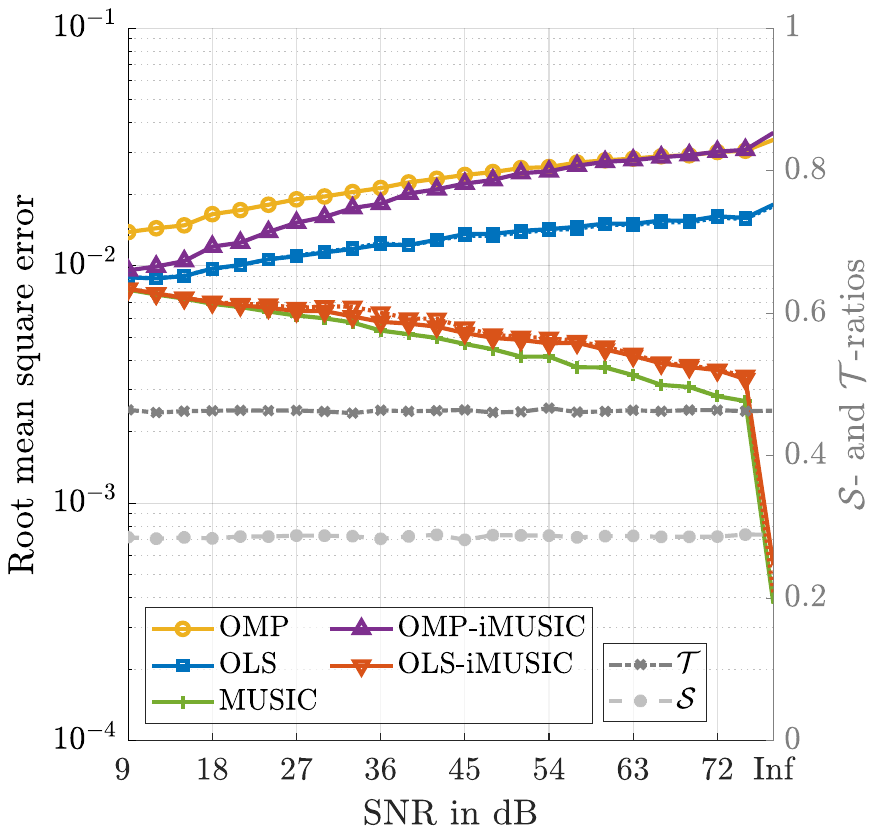}\label{fig:RMSE_vs_SNR}
}
\hfill
\subfloat[]{
\includegraphics[trim=3.4cm 9.5cm 4.3cm 9.91cm,width = 0.3\textwidth]{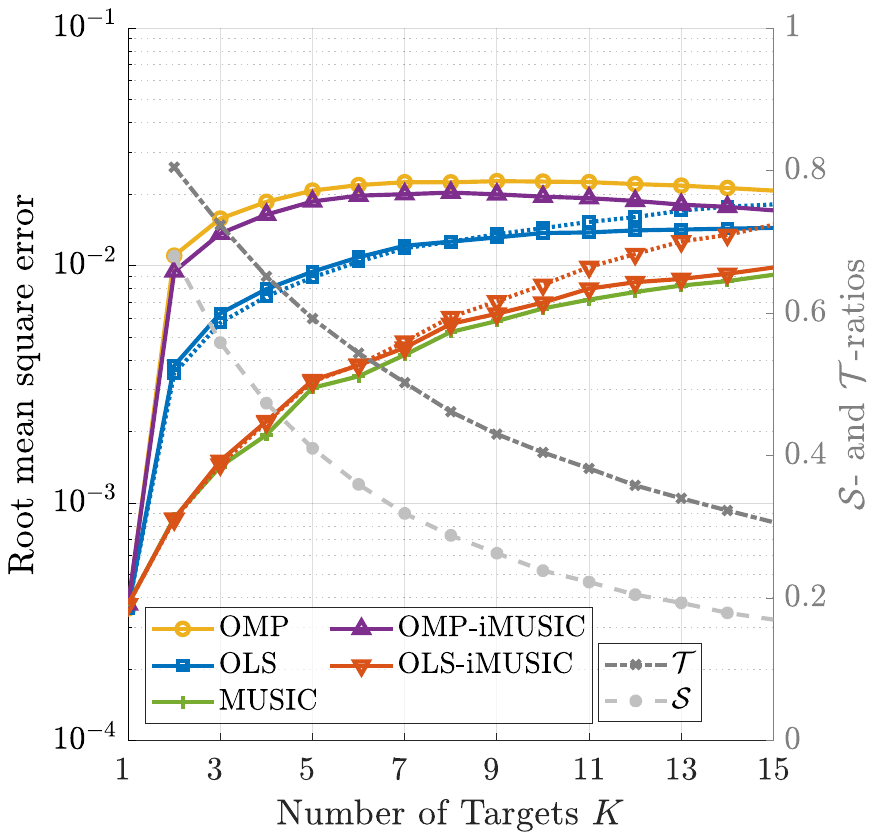}\label{fig:RMSE_vs_K}
}
\hfill
\subfloat[]{
\includegraphics[trim=3.4cm 9.5cm 4.3cm 9.91cm,width = 0.3\textwidth]{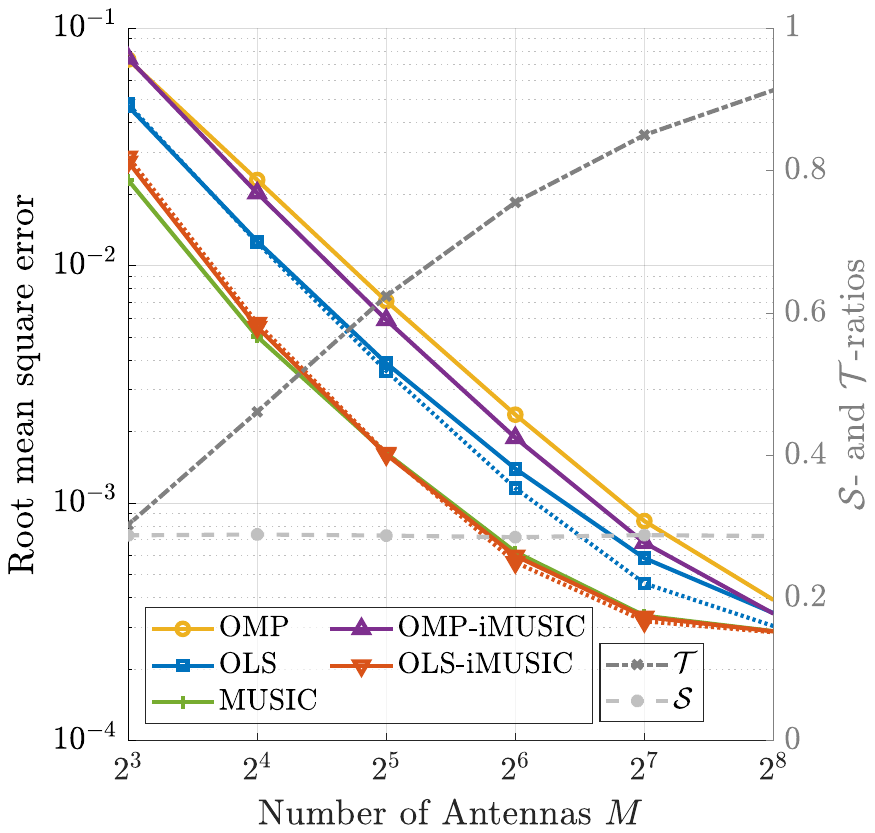}\label{fig:RMSE_vs_N}
}
\end{center}
\vspace*{-0.3\textwidth}
\hspace{-.01\textwidth} (a) \hspace{.305\textwidth} (b) \hspace{.305\textwidth} (c) \hfill
\vspace*{0.27\textwidth}
\caption{\Gls{rmse} versus (a) \gls{snr}, (b) number of targets $\TgtAll$, and (c) number of antennas $\SRxAll$. Solid lines use the eigenvalue-based \gls{aic} stopping criterion, whereas dotted lines use the hybrid criterion. The $\CorrRatio$ and $\ResRatio$ diagnostics are also plotted (dashed light gray and dash-dotted dark gray, respectively).}
\label{fig:RMSE_comparison}
\end{figure*}
We now turn to precision and evaluate it through the \gls{rmse} between true normalized \glspl{doa} and estimated normalized \glspl{doa}.
To ensure a fair comparison, we compute \gls{rmse} only on targets that are hits for all methods.
In other words, we evaluate precision only on the subset of targets that are detected by all methods, which are typically the easiest ones. 
Thus, the precision analysis in Figure~\ref{fig:RMSE_comparison} focuses on the accuracy of peak localization, rather than on the ability to detect difficult targets.

Similarly to the detection analysis, \figurename~\ref{fig:RMSE_vs_SNR} shows that precision improves with \gls{snr} for \gls{music} and \gls{olsimusic}.
By contrast, for the remaining methods, precision degrades as \gls{snr} increases because more difficult targets are counted as hits, with less accurate peak localization.
This trend is consistent with the lack of super-resolution in \gls{omp} and \gls{ols}, and with its only partial exploitation in \gls{ompimusic}.
The same behavior persists in the noiseless-limit regime, where \gls{snr} tends to infinity.

Consistently, \figurename~\ref{fig:RMSE_vs_K} shows that precision degrades as $\TgtAll$ increases, in line with decreasing $\CorrRatio$ and $\ResRatio$.
All methods coincide at $\TgtAll=1$, as expected for single-target estimation.
As $\TgtAll$ grows, the precision gain of \gls{olsimusic} decreases, unlike the detection trend where gains become more visible at high $\TgtAll$.
This difference arises because precision is computed only over targets detected by all methods, which are typically the easiest ones.

\figurename~\ref{fig:RMSE_vs_N} further shows that precision improves with the number of antennas $\SRxAll$, as expected from the increase in $\ResRatio$.
As $\ResRatio$ increases, all methods improve similarly until they reach the resolution limit imposed by the discrete grid $\Grid_\Aoa$ (spacing $\sim 10^{-3}$).
This behavior contrasts with the gain obtained by increasing $\CorrRatio$ through larger $\SubAll$.
Specifically, increasing $\SRxAll$ (and thus improving $\ResRatio$) benefits all methods, whereas increasing $\CorrRatio$ benefits \gls{music}-based and eigenvalue-criterion-based methods more strongly because they are more sensitive to inter-target signal correlation.

Overall, \gls{music} and \gls{olsimusic} outperform \gls{omp}, \gls{ompimusic}, and \gls{ols} in precision.
This behavior is expected because subspace methods such as \gls{music} provide super-resolution.
By contrast, \gls{ompimusic} does not reach \gls{music}-level precision because the greedy \gls{omp} framework does not solve the \gls{ml} least-squares problem at each iteration, which limits peak-localization accuracy.
Therefore, \gls{olsimusic} tends to match the near-optimal precision of \gls{music} by solving a least-squares problem at each iteration, whereas \gls{ompimusic} remains closer to \gls{omp}.

\subsubsection{\textbf{Timing Performance}}
\begin{figure*}[!t]
\captionsetup[subfigure]{labelformat=empty, labelsep=none, justification=raggedright, singlelinecheck=false}

\hspace{0.05\textwidth}
\begin{center}
\subfloat[]{
\includegraphics[trim=3.4cm 9.5cm 4.3cm 9.91cm,width = 0.3\textwidth]{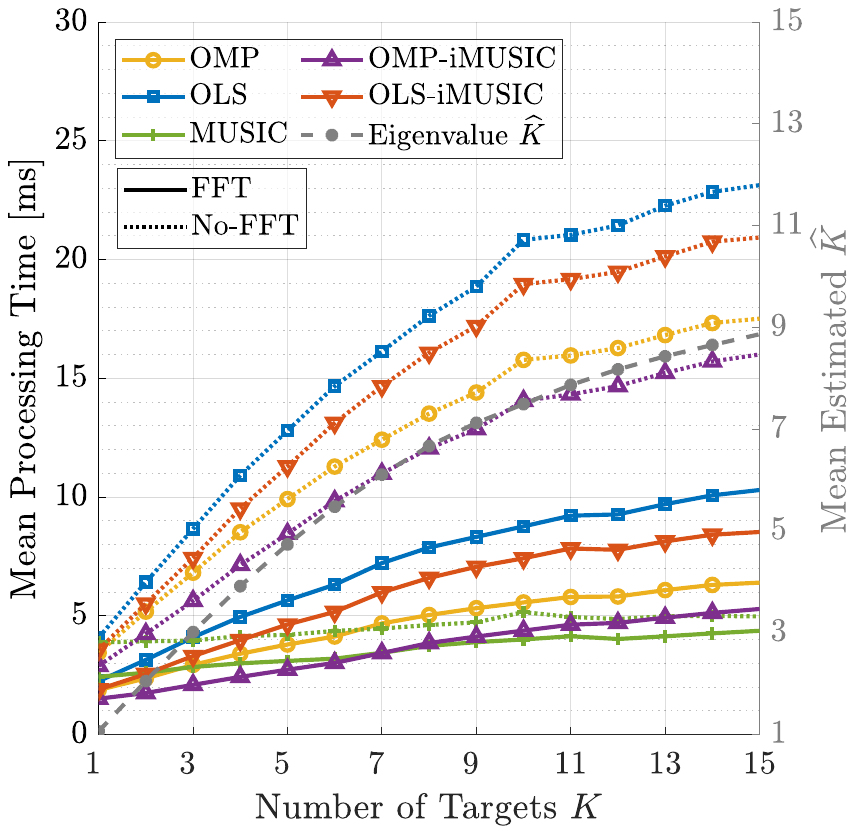}\label{fig:Timing_vs_K}
}
\hfill
\subfloat[]{
\includegraphics[trim=3.4cm 9.5cm 4.3cm 9.91cm,width = 0.3\textwidth]{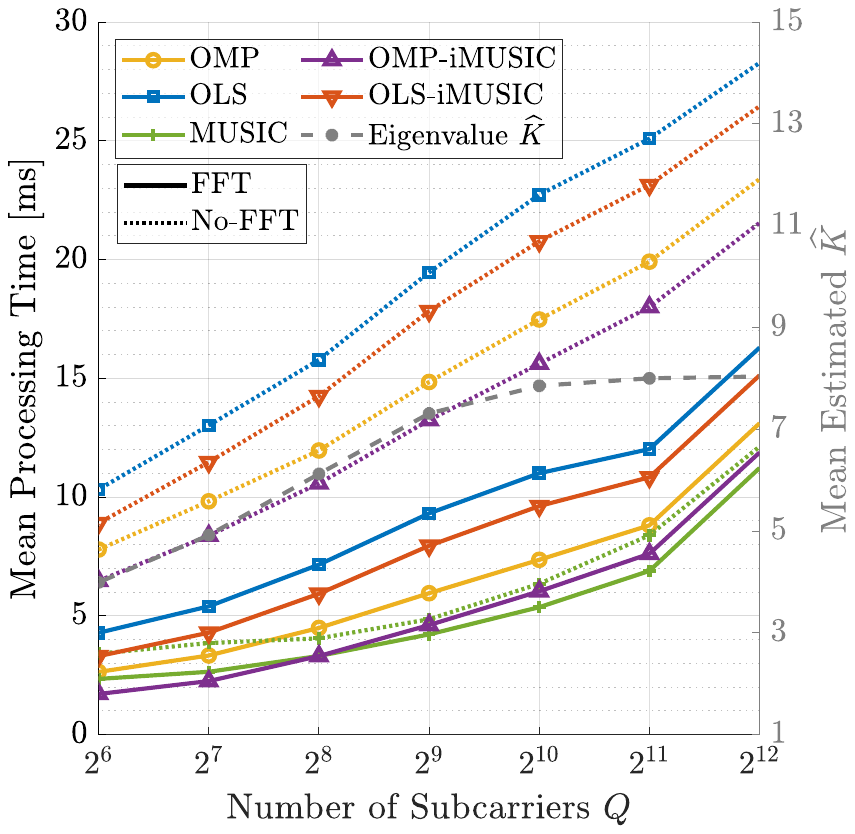}\label{fig:Timing_vs_Q}
}
\hfill
\subfloat[]{
\includegraphics[trim=3.4cm 9.5cm 4.3cm 9.91cm,width = 0.3\textwidth]{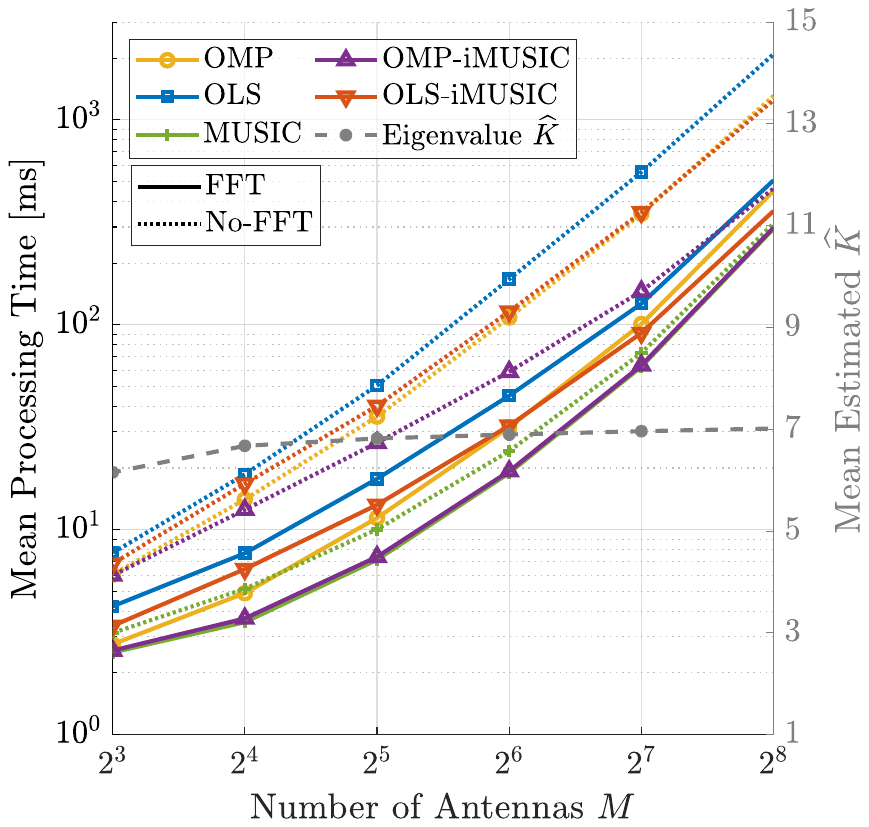}\label{fig:Timing_vs_N}
}
\end{center}
\vspace*{-0.3\textwidth}
\hspace{-.01\textwidth} (a) \hspace{.305\textwidth} (b) \hspace{.305\textwidth} (c) \hfill
\vspace*{0.27\textwidth}
\caption{Mean processing time versus (a) number of targets $\TgtAll$, (b) number of subcarriers $\SubAll$, and (c) number of antennas $\SRxAll$. Solid lines denote \gls{fft}-accelerated implementations, whereas dotted lines denote non-accelerated implementations. In (c), the y-axis is logarithmic. The mean number of targets estimated by the eigenvalue-based criterion is also shown (dashed gray).}
\label{fig:Timing_comparison}
\end{figure*}
We next evaluate timing performance by averaging the mean processing time of each algorithm, with and without \gls{fft} acceleration, across all simulations.
As discussed in Section~\ref{sec:Complexity Analysis}, we implement \gls{omp} and \gls{ols} using $\SqrtMat{\SCovMat}_{\TgtInd}$ rather than the observation matrix $\SRxSigMat_{\TgtInd}$ because $\ObsAll\SubAll \gg \SRxAll$ in the considered scenario.
This choice substantially reduces processing time.
\figurename~\ref{fig:Timing_comparison} reports these results, where solid lines denote \gls{fft}-accelerated implementations and dotted lines denote non-accelerated implementations.
To keep the comparison consistent, we show only the eigenvalue-based \gls{aic} criterion, so that all methods are compared under the same target-count estimate.

The observed timing trends are consistent with the theoretical complexity summary in Table~\ref{tab:complexity_comparison}.
In \figurename~\ref{fig:Timing_vs_K}, iterative-method runtime scales approximately linearly with $\TgtAll$.
However, this increase becomes weaker at high $\TgtAll$ because the eigenvalue-based \gls{aic} criterion saturates for large target counts \cite{willameStopping}.
By comparison, \gls{music} is less sensitive to $\TgtAll$ because this parameter only changes subspace dimensions through $\TgtAll_\SRxAll$.

In \figurename~\ref{fig:Timing_vs_Q}, the runtime of all methods scales linearly with the number of subcarriers $\SubAll$.
This trend is mainly driven by the increase in $\CorrRatio$ as $\SubAll$ grows, which improves the eigenvalue-based \gls{aic} criterion and thereby increases the number of detected targets across methods.

\figurename~\ref{fig:Timing_vs_N} shows that runtime scales quadratically with the number of antennas $\SRxAll$.
Here, the number of targets detected by the eigenvalue-based \gls{aic} criterion changes only marginally with $\SRxAll$.
Therefore, the dominant effect is the cost of evaluating the selection-step objective function, which scales as $\SRxAll^2$ (or $\SRxAll\log(\SRxAll)$ with \gls{fft} acceleration) for \gls{omp}, \gls{ols}, and \gls{ompimusic}, and as $\SRxAll$ (or $\log(\SRxAll)$ with \gls{fft} acceleration) for \gls{music} and \gls{ompimusic}.
Additionally, the \gls{evd} cost increases with $\SRxAll$, which further contributes to the runtime increase of all methods.

Overall, in agreement with the complexity analysis, \gls{gimusic} methods are less computationally demanding than \gls{omp} and \gls{ols}, respectively.
In addition, \gls{fft} acceleration significantly reduces processing time across all methods, with the largest gains for iterative methods, making them more suitable for real-time operation and competitive with \gls{music}.
At low $\TgtAll$, \gls{ompimusic} can even be faster than \gls{music}, because \gls{music} requires a more expensive peak-detection stage, whereas \gls{ompimusic} uses a simple maximum search.

\subsubsection{\textbf{\gls{giwmusic} Performance}}
For visual clarity, the curves of \gls{giwmusic} methods in Figures~\ref{fig:Jstat_comparison}, \ref{fig:RMSE_comparison}, and \ref{fig:Timing_comparison} are omitted.
For detection and precision, the curves of \gls{ompiwmusic} and \gls{olsiwmusic} are superimposed on those of \gls{omp} and \gls{ols}, respectively.
This overlap indicates that the low-rank approximation introduced by \gls{giwmusic} does not degrade performance.
For timing, their curves are superimposed on \gls{ompimusic} and \gls{olsimusic}, respectively, because weighting introduces a negligible computational overhead relative to unweighted variants.
Therefore, \gls{giwmusic} methods are fast alternatives to \gls{omp} and \gls{ols}, respectively, without performance degradation.
However, in the considered scenario, they are outperformed by their non-weighted \gls{gimusic} counterparts in terms of precision and detection performance.

\section{Conclusion} \label{sec:conclusion}
In this paper, we presented novel \gls{gimusic} algorithms for \gls{doa} estimation.
These methods integrate the \gls{music} subspace approach into the iterative greedy frameworks of \gls{omp} and \gls{ols}.
Unlike existing iterative \gls{music} approaches, the proposed methods adopt a principled greedy approach and avoid recomputing a new \gls{evd} at each iteration, thereby reducing computational cost.
To further lower complexity, we also introduced \gls{fft}-based acceleration of the \gls{gimusic} selection steps when a \gls{ula} is used.

Through numerical simulations, we demonstrated that the proposed \gls{ompimusic} and \gls{olsimusic} methods  outperform their respective \gls{omp} and \gls{ols} baselines in terms of detection and precision, while also improving upon classical \gls{music}.
In parallel, the complexity analysis showed that the proposed methods are less demanding than their \gls{omp} and \gls{ols} counterparts, with only limited additional overhead relative to \gls{music}.
The \gls{giwmusic} variants provide low-rank approximations of the corresponding \gls{omp} and \gls{ols} selection steps, without performance degradation and with reduced computational cost.
However, they are outperformed by their unweighted \gls{gimusic} counterparts in the considered scenario.
As a final contribution, we interpreted these outcomes using two diagnostic metrics, namelyn the steering-vector correlation metric ($\ResRatio$) and the signal correlation metric ($\CorrRatio$), which provide scenario-agnostic insight into expected performance beyond the passive-radar setting considered here.

In conclusion, these results indicate that \gls{ompimusic} and \gls{olsimusic} offer an effective balance between estimation performance and complexity for real-time applications.
The choice between the two variants can therefore be guided by the desired performance--complexity trade-off.
Building on this foundation, future work will extend the proposed \gls{gimusic} framework to other radar configurations, including \gls{mimo} and multistatic systems.

\balance
\bibliographystyle{IEEEtran}
\bibliography{IEEEabrv,References}
\end{document}